\documentclass[letterpaper,twocolumn,10pt]{article}
\usepackage{usenix2019_v3}

\usepackage[small,compact]{titlesec}

\usepackage{tikz}
\usepackage{amsmath}

\usepackage{filecontents}

\usepackage{microtype}
\usepackage[english]{babel}
\usepackage{blindtext}
\usepackage{subfig}
\usepackage{graphicx}	
\usepackage{multirow}
\usepackage{algorithm}
\usepackage[normalem]{ulem} 
\usepackage{hyperref}
\hypersetup{
 colorlinks=true,
 linkcolor=blue,
 filecolor=magenta,      
 urlcolor=cyan,
 breaklinks=true,
}
\usepackage{comment}

\usepackage[english]{babel}

\usepackage{booktabs}
\usepackage{subfig}
\usepackage{graphicx}
\usepackage[normalem]{ulem}
\usepackage{array}
\usepackage{algorithm}
\usepackage[noend]{algpseudocode}
\usepackage{bold-extra}
\usepackage{tikz}
\usepackage{amsmath}
\usepackage{xspace}
\usepackage{filecontents}
\usepackage{listings}
\usepackage{etoolbox}
\usepackage{longtable}
\usepackage{lipsum}

\usepackage{pifont}
\usepackage{multirow}
\usepackage{paralist}
\usepackage{enumitem}

\usepackage{url}

\usepackage{ifthen}
\usepackage{xcolor}
\usepackage{blindtext}
\usepackage{algorithm}
\usepackage{graphicx}
\usepackage{amsmath}
\usepackage[noend]{algpseudocode}
\usepackage{subfig}
\usepackage{xspace}
\usepackage{amsthm}
\usepackage{balance}

\usepackage{graphicx}
\usepackage{footnote}
\usepackage{comment}
\usepackage{enumitem,kantlipsum}
\makesavenoteenv{tabular}
\makesavenoteenv{table}
\usepackage{listings}
\usepackage{colortbl}

\newcommand{\exclude}[1]{}
\newcommand{\showComments}{yes}
\newcommand{\note}[2]{
    \ifthenelse{\equal{\showComments}{yes}}{\textcolor{#1}{#2}}{}
}

\makeatletter
\let\OldStatex\Statex
\renewcommand{\Statex}[1][3]{%
  \setlength\@tempdima{\algorithmicindent}%
  \OldStatex\hskip\dimexpr#1\@tempdima\relax}
\makeatother
\algnewcommand{\LeftComment}[1]{\OldStatex \(\triangleright\) #1}
\algnewcommand{\LineComment}[1]{\OldStatex \(\triangleright\) #1}

\newcommand{\para}[1]{{\textbf{{#1}}}}
\newcommand{\name}{{\sc{Cassini}}\xspace}
\newcommand{\namebf}{\bfseries{\scshape{Cassini}}\xspace}

\newcommand{\on}{Up\xspace}
\newcommand{\off}{Down\xspace}
\newcommand{\affinity}{Affinity\xspace}
\newcommand*\mycircle[1]{%
\scalebox{0.9}{\begin{tikzpicture}[baseline=-3pt]
  \node[draw=none,circle,inner sep=0.5pt, fill=black] {\textcolor{white}{\textsf{\textbf{#1}}}};
\end{tikzpicture}}}

\algnewcommand{\IOComment}[1]{\OldStatex \(\triangleright\) #1}
\algnewcommand{\firstLeftComment}[1]{\OldStatex \(\indent\triangleright\) #1}
\algnewcommand{\secondLeftComment}[1]{\OldStatex \(\indent\indent\triangleright\) #1}
\algnewcommand{\thirdLeftComment}[1]{\OldStatex \(\indent\indent\indent\triangleright\) #1}
\definecolor{LightCyan}{rgb}{0.88,1,1}
\definecolor{celadon}{rgb}{0.67, 0.88, 0.69}
\providecommand{\para}[1]{\smallskip\noindent\textbf{#1} }
\usepackage{caption}
\begin{document}
\title{\namebf: Network-Aware Job Scheduling in Machine Learning Clusters}
\pagenumbering{arabic}

\author{Sudarsanan Rajasekaran$^\dagger$ \qquad Manya Ghobadi$^\dagger$ \qquad Aditya Akella$^\ddagger$ \\\\
\vspace{0.3cm}
$^\dagger$Massachusetts Institute of Technology \qquad
$^\ddagger$UT Austin
}

\setcounter{page}{1}
\maketitle

\begin{abstract}

We present \name, a network-aware job scheduler for machine learning (ML) clusters. \name introduces a novel geometric abstraction to consider the communication pattern of different jobs while placing them on network links. To do so, \name uses an affinity graph that finds a series of time-shift values to adjust the communication phases of a subset of jobs such that the communication patterns of jobs sharing the same network link are interleaved with each other.
Experiments with 13 common ML models on a 24-server testbed demonstrate that compared to the state-of-the-art ML schedulers, \name improves the average and tail completion time of jobs by up to 1.6$\times$ and 2.5$\times$, respectively. Moreover, we show that \name reduces the number of ECN marked packets in the cluster by up to 33$\times$.

\end{abstract}
\section{Introduction}

The ever-growing increase in dataset and model sizes of deep learning has created a massive demand for efficient GPU clusters. Several studies have  demonstrated that as the number of GPUs increases, the communication overhead of distributed Machine Learning (ML) training workloads quickly takes up a significant portion of training iteration time~\cite{pipedream, sip-ml, gpipe, horovod, dcscale_ml, mudigere2021highperformance, network_evolution}. Yet, state-of-the-art ML schedulers, tend to ignore the communication pattern of ML training jobs when placing workers on GPUs.

In this paper, we develop a simple but effective approach, called \name, that integrates with existing ML schedulers to allow them to efficiently place multiple ML jobs on network links while minimizing the chances of network congestion. Our approach requires no special support, such as reservations and priorities, from switches/NICs and does not require any changes to the congestion control protocol.

\begin{figure*}[t]
    \includegraphics[width=1.0\textwidth]{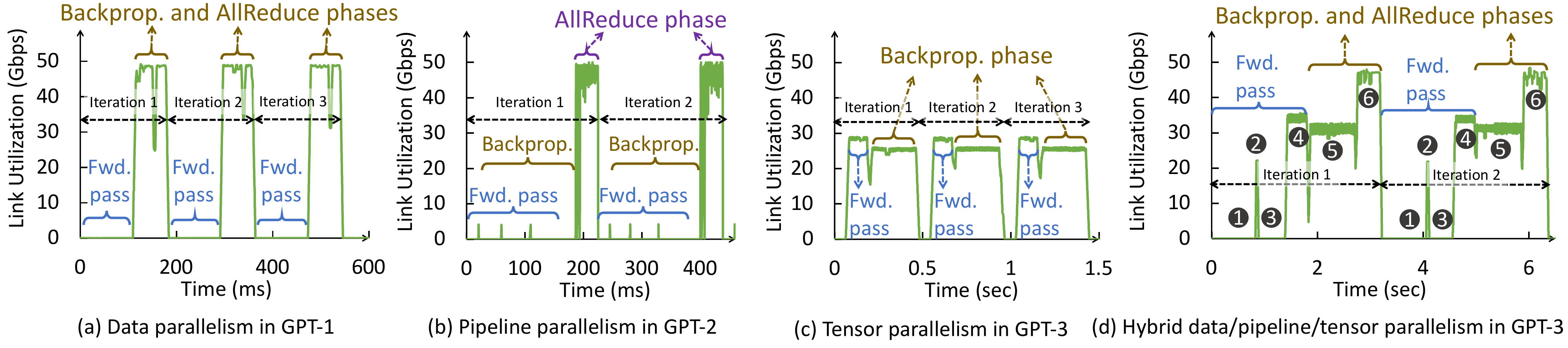}
    \caption{The traffic pattern of different parallelization strategies when training GPT-1, GPT-2, and GPT-3 models.}
    \vspace{-0.3cm}
\label{fig:on_off_pattern_parallelization_strategy}
\end{figure*}

We demonstrate that for a specific combination of jobs, introducing a small time-shift to delay the start of one of the iterations enables \name to interleave the computation and communication patterns of different jobs, thereby improving the training time. We refer to such combinations of jobs as \textit{compatible} and develop \name as a pluggable module to augment prior ML schedulers to consider a novel {\em compatibility metric} when determining where to place jobs and control how jobs compete on network links. 

Augmenting ML schedulers to take links and servers into account is inherently challenging because jobs are likely to traverse multiple links and may compete with different jobs on different links. 
To address this challenge, we propose a {\em geometric abstraction} that leverages the periodic communication pattern of Deep Neural Network (DNN) training workloads. The key idea of our abstraction is to ``roll'' time around a circle whose perimeter is proportional to the training iteration time of ML jobs. To determine the compatibility score of two (or more) jobs on a link, \name places each job on its corresponding circle and overlays the circles on top of each other. It then uses an optimization formulation to rotate the circles into a position that maximizes interleaving. The rotation angle of each job corresponds to a time-shift value to delay the start of the next immediate training iteration to achieve compatibility.

Looking beyond a single link and extending to jobs running across a topology, we generalize the geometric abstraction to cluster-level by introducing a bipartite \affinity graph where vertices are a subset of jobs and links. An edge in the \affinity graph indicates a job is traversing a link. We then use a new graph traversal algorithm to find unique time-shifts for all jobs while maintaining their compatibility on all links. Using our geometric abstraction and \affinity graph, we augment Themis~\cite{themis} and Pollux~\cite{pollux} with $\approx$1000 lines of code.

To evaluate \name, we build a testbed with 24 servers, each with one NVIDIA A100 GPU~\cite{a100} and one 50~Gbps RDMA NIC. Our experiments with 13 representative DNN models (VGG11~\cite{vgg11}, VGG16~\cite{vgg16}, VGG19~\cite{vgg19}, ResNet50~\cite{resnet}, WideResNet101~\cite{wideresnet}, BERT~\cite{bert}, RoBERTa~\cite{roberta}, XLM~\cite{xlm}, CamemBERT~\cite{camembert}, GPT-1~\cite{gpt_1}, GPT-2~\cite{gpt_2}, GPT-3~\cite{gpt_3}, and DLRM~\cite{dlrm}) show that \name improves the tail completion time of jobs by up to 2.2$\times$ and 2.5$\times$, compared to Themis~\cite{themis} and Pollux~\cite{pollux}, respectively. 
Moreover, we show that \name reduces the number of ECN marked packets in the cluster by up to 33$\times$. 
\section{Background and Motivation}
\label{sec:motivation}

\subsection{Distributed DNN Training Traffic Pattern}
\label{sec:traffic_pattern}

\name is designed for large GPU clusters with hundreds of training jobs distributed with data, pipeline, and/or model parallel training paradigms. 
To this end, we study the impact of different parallelization strategies on network demand using a series of measurements. Each server in our testbed has one A100 GPU and one ConnectX-5 Mellanox RDMA NIC with 50~Gbps capacity. In all our experiments, we choose batch sizes such that the GPU utilization is higher than 80\% and that intra-job pipelining is enabled.

\para{Data parallelism.} In data parallel training, the DNN model is copied into the memory of each GPU while the dataset is distributed across them. Figure~\ref{fig:on_off_pattern_parallelization_strategy}(a) shows the communication pattern of a GPT-1~\cite{gpt_1} model (12 layers, 9~GB memory) distributed across four GPU servers using data parallelism. The figure shows the traffic pattern of three back-to-back training iterations. Each iteration contains a forward pass with near-zero network demand, followed by a period of high utilization corresponding to the backpropagation and AllReduce phases.

\para{Model/Pipeline parallelism.}  
In model parallel training, the DNN model is partitioned across workers~\cite{NIPS2014, pmlr-v80-jia18a}, and parts of the DNN model are computed on different workers. The two common techniques for model parallelism are tensor parallelism and pipeline (or layer) parallelism~\cite{demystifying_ml}. In pipeline parallelism, the model is partitioned vertically at the layer boundaries~\cite{pipedream, gpipe}. Figure~\ref{fig:on_off_pattern_parallelization_strategy}(b) shows the communication pattern of a GPT-2~\cite{gpt_2} model (24 layers, 27~GB memory) distributed across two servers using pipeline parallelism. We partition the model vertically in half (i.e., server$_1$ contains layer$_1$ to layer$_{12}$ and server$_2$ contains layer$_{13}$ to layer$_{24}$) and use PipeDream's approach~\cite{pipedream} to divide the batch size into three minibatches. The three small communication peaks  during the forward pass correspond to the activation parameters of these three minibatches. The heavy communication demand following the peaks corresponds to the AllReduce operation between the embedding layers in the model.

\begin{figure*}[t]
\centering
\includegraphics[width=0.9\textwidth]{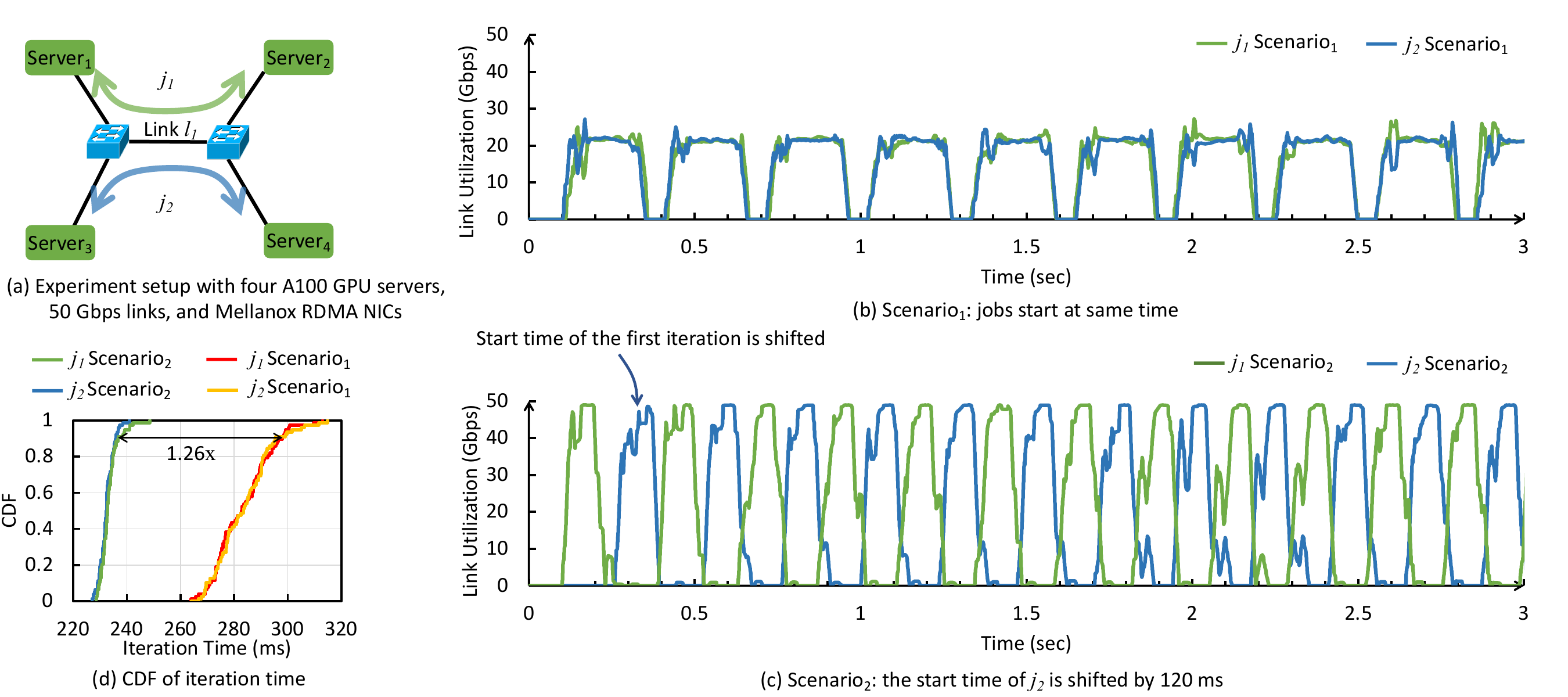}
\caption{Impact of interleaving the \on-\off phases of two VGG19 jobs sharing link $l_1$.}
\vspace{-0.4cm}
\label{fig:compatible_jobs}
\end{figure*}

\para{Model/Tensor parallelism.} Another variant of model parallel training is tensor parallelism~\cite{mesh_tensorflow, megatron_modelparallel}.
Tensor parallelism techniques partition the model horizontally such that different tensors are distributed across workers~\cite{sagemaker, tesseract}. Figure~\ref{fig:on_off_pattern_parallelization_strategy}(c) shows the communication pattern of a GPT-3~\cite{gpt_3} model (96 layers, 35~GB memory) distributed across two servers using tensor parallelism. We partition the model horizontally in half, where each server contains half of all the layers.
The figure shows that both forward and backpropagation phases introduce roughly 25~Gbps traffic followed by a short period of near-zero network demand during data loading.

\para{Hybrid data/pipeline/tensor parallelism.} Today's DNN training systems tend to use a hybrid of data/pipeline/tensor parallelism to train large DNN models~\cite{bert_blog, sip-ml, model-data-parallel, topoOpt}. Figure~\ref{fig:on_off_pattern_parallelization_strategy}(d) shows the communication pattern of a GPT-3~\cite{gpt_3} model (96 layers, 155~GB memory) distributed across eight servers using hybrid data/pipeline/tensor parallelism. We use pipeline parallelism to partition the model's layers vertically into two parts. Then, we divide the layers in each partition horizontally to obtain a total of four partitions. Next, we assign each of these four partitions to a server. Finally, we replicate the same process across another group of four servers and use data parallelism to distribute the data between these two groups of four servers. The figure shows the communication demand of the forward, backpropagation, and AllReduce phases where each phase has a different network demand.

\para{Key takeaways.} We repeat the above experiments using common DNN models, such as BERT~\cite{bert}, DLRM~\cite{dlrm}, WideResNet101~\cite{wideresnet},  RoBERTa~\cite{roberta}, VGG~\cite{vgg} and observe similar traffic patterns. Our key takeaways are: ($i$) the network demand repeats itself across all iterations, as long as the training hyper-parameters remain the same; ($ii$) the network demand of an iteration may consist of multiple \on and \off phases. The exact magnitude of the network demand during these \on and \off phases depends on the parallelization strategy and hyper-parameters. For instance, Figure~\ref{fig:on_off_pattern_parallelization_strategy}(d) shows each training iteration has six \on-\off phases, labeled as \mycircle{1} to \mycircle{6}. Section~\ref{sec:geometric} describes \name's approach to capture the duration and bandwidth of \on-\off phases.

\subsection{Interleaving the \on and \off Phases}
\label{sec:interleaving_potential}

\name's goal is to augment ML schedulers to consider the traffic demand of training jobs when making scheduling decisions. In particular, given the key takeaways in the previous section, we aim to interleave the bandwidth demand of \on and \off phases of different jobs to leverage the periodic network demand of distributed DNN training jobs.

To demonstrate the power of \on-\off network demand interleaving, consider two data parallel training jobs, $j_1$ and $j_2$, as shown in Figure~\ref{fig:compatible_jobs}(a).  Each job has one \on and one \off phase at every training iteration. We run each job for 1,000 iterations under two scenarios. In the first scenario, two VGG19~\cite{vgg} jobs start simultaneously and share $l_1$ fairly. The communication uses the RDMA-based DCQCN congestion control algorithm~\cite{dcqcn}. Figure~\ref{fig:compatible_jobs}(b) shows that both jobs achieve roughly 22~Gbps bandwidth (i.e., half of $l_1$'s capacity). In the second scenario, shown in Figure~\ref{fig:compatible_jobs}(c), we interleave the \off phase of $j_1$ with the \on phase of $j_2$ and vice versa, by shifting the start time of $j_2$ by 120~ms (Section~\ref{sec:geometric} describes how we obtained this value). In this scenario, the jobs do not compete for bandwidth during their respective \on phases, giving both jobs the entire available bandwidth.  Figure~\ref{fig:compatible_jobs}(d) plots the CDF of training iteration times for both scenarios demonstrating that scenario$_2$ accelerates the 90$^{th}$ percentile tail iteration time of both jobs by 1.26$\times$. 

Perfectly interleaving the \on and \off phases of different jobs is not always possible.
For instance, when BERT~\cite{bert} and VGG19~\cite{vgg} models share a link, no suitable time-shift can achieve perfect interleaving. But when WideResNet101~\cite{wideresnet} and VGG16~\cite{vgg} share a link, shifting VGG16 by 150~ms enables perfect interleaving.
Instead of relying on perfectly matching \on and \off phases of jobs, we define a metric called \textit{compatibility rank} that captures the potential degree of interleaving across jobs sharing the network. Section~\ref{sec:geometric}  describes a novel technique to determine the compatibility rank and the amount of required time-shift to achieve it.

\section{Geometric Abstraction}
\label{sec:geometric}

\begin{figure*}[t]
    \begin{minipage}{0.75\textwidth}
    \includegraphics[width=1.0\textwidth]{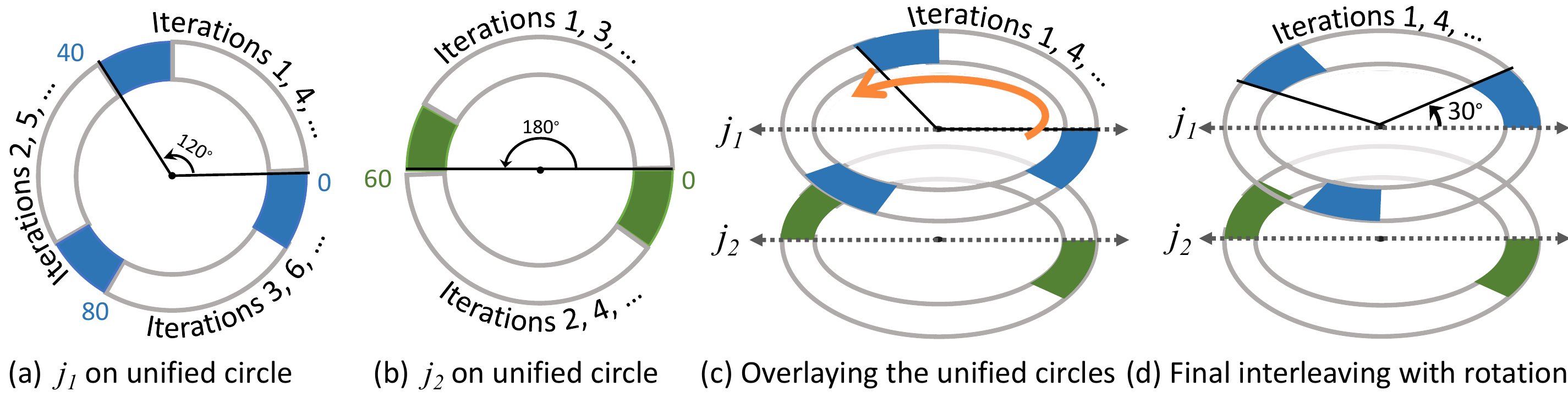}
    \caption{\name's unified circles for jobs with different iteration times.}
    \label{fig:overlapability}
    \end{minipage}
    \hspace{0.25cm}
    \begin{minipage}{0.2\textwidth}
    \includegraphics[width=1\textwidth]{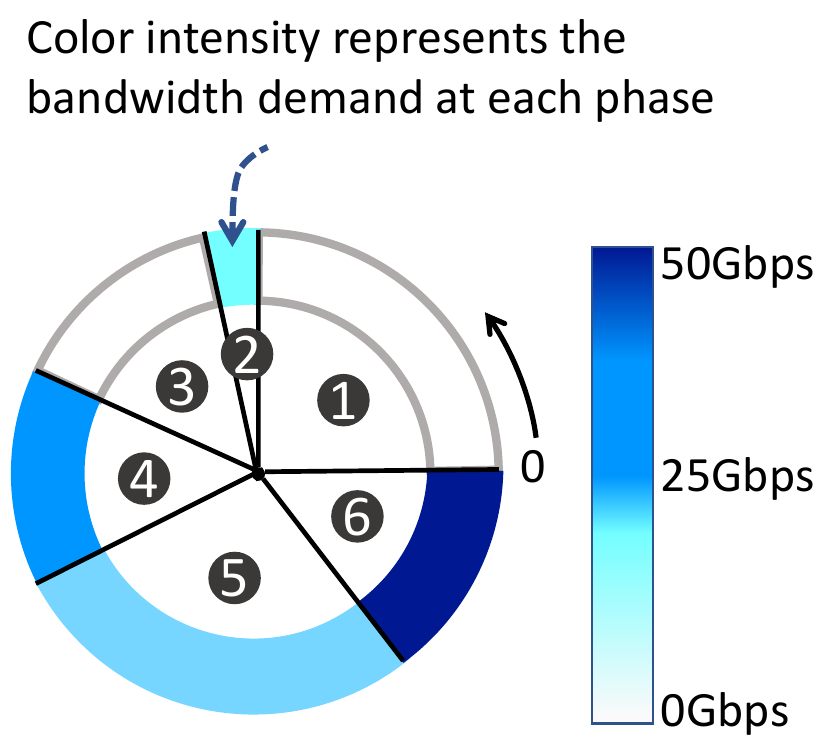}
    \caption{Geometric circle for the job in Fig.~\ref{fig:on_off_pattern_parallelization_strategy}(d).}
    \vspace{-0.4cm}
    \label{fig:model_parallel_geometric}
    \end{minipage}
\end{figure*}

Consider a time-series representation of the network demand for a job running in a dedicated cluster with no congestion. As shown in Section~\ref{sec:motivation}, different training jobs have different \on and \off patterns but the duration and bandwidth demand of the same job remain more or less the same across training iterations. The key idea of our abstraction is to \textit{roll} time around a circle whose \textit{perimeter} is equal to the  iteration time of a job. Consequently, the \on-\off phases of all iterations will appear on approximately the same angles of the circle.

Figure~\ref{fig:example_geometric}(a) illustrates the time-series network demand of a data parallel VGG16 training job with a training iteration time of 255~ms. Figure~\ref{fig:example_geometric}(b) shows a circle with perimeter 255 units where the time-series data is plotted around it. The figure demonstrates that the \on and \off phases of different iterations cover the same angles of the circle. Our geometric abstraction captures this property, as shown in Figure~\ref{fig:example_geometric}(c). The perimeter of the circle is the iteration time, set to 255 units. The \off phase spans 141 units, represented by the uncolored arc with 200$^\circ$ angle, starting at 0$^{\circ}$, on the x-axis. The \on phase represented by the colored arc occupies the remainder of the circle.

\para{Rotate the circle to interleave \off and \on phases of different jobs.} To determine the compatibility score of two (or more) jobs on a link, we place each job on its corresponding circle and overlay the circles on top of each other. Congestion happens when the total bandwidth demand of a particular angle is higher than the capacity of the link, as shown in Figure~\ref{fig:unfairness_to_geometric}(a). To find the best interleaving, we rotate the circles to a position where the summation of the bandwidth demands is less than the capacity of the link for all angles in the circle, as shown in Figure~\ref{fig:unfairness_to_geometric}(b). If such a rotation is found, the jobs are fully  compatible. 

\para{Capturing jobs with different iteration times.} The above technique only works when the perimeter of circles is the same. To generalize to the case where jobs have different iteration times, we place each job on a  \textit{unified circle} whose perimeter is equal to the Least Common Multiple (LCM) of the iteration time of all jobs competing on the link. For instance, consider two jobs $j_1$ and $j_2$ competing on a bottleneck link with iteration times 40~ms and 60~ms, respectively. To determine the compatibility score of the two jobs, we place them on a circle with a perimeter equal to $LCM(40, 60) = 120$ units. Figure~\ref{fig:overlapability}(a) shows $j_1$ on this unified circle. As the perimeter of the circle is 3$\times$ $j_1$'s iteration time, there are three periods of \on and \off phases in the figure. Similarly, Figure~\ref{fig:overlapability}(b) shows $j_2$ on the unified circle. We then overlay the unified circles on top of each other (shown in Figure~\ref{fig:overlapability}(c)) and rotate the circles to determine the compatibility score. Figure~\ref{fig:overlapability}(d) shows that by rotating $j_1$ by $\Delta=30^{\circ}$ counter-clockwise, the sum of bandwidth demands on all angles of the unified circles is lower than the link capacity,
giving these two jobs a compatibility score of 1 (i.e., fully compatible). 

\para{Capturing the bandwidth demand of model parallel training jobs.} For clarity of presentation, the examples in this section contained data parallel training jobs with one \on and one \off phase during each iteration. \name's geometric abstraction is generic and can capture more complex traffic patterns induced by various parallelization paradigms. Consider the communication pattern of the GPT-3 model with hybrid data/pipeline/tensor parallelism shown in Figure~\ref{fig:on_off_pattern_parallelization_strategy}(d). Here, GPT-3's communication pattern consists of six \on-\off phases with different durations and bandwidth demands. The geometric circle of this job contains six colored arcs where the length and color intensity of each arc corresponds to the duration and bandwidth demand of each \on-\off phase of the model, shown in Figure~\ref{fig:model_parallel_geometric}. Next, we formalize our geometric representation and show an optimization formulation that uses the geometric abstraction to find rotation angles to interleave the \on-\off phases of  multiple jobs sharing a link, irrespective of the parallelization strategy.

\para{Finding rotation angles.} Once jobs are placed on their unified circles, \name uses an optimization formulation, shown in Table~\ref{tab:formulation_phase_2}, to find the best angle of rotation for jobs to maximize their compatibility.

\para{Optimization input.} The input is a set of ML jobs $J^l=\{j\}$  competing on a link $l$. We profile each job $j$ to construct its unified circle, denoted by unified\_circle$_j$. The perimeter of the unified circle is the LCM of the iteration times of all jobs $j \in J^l$.
The data structure of unified\_circle$_j$ contains a series of bandwidth demands, $bw\_circle_{j}(\alpha)$, where $\alpha \in [0, 2\pi]$ identifies an arc of the circle that corresponds to an \on or \off phase in the communication pattern.
The total capacity of link $l$ is denoted by $C^l$.

\para{Optimization objective and output.} 
The optimization goal is to overlay the unified circles of each job and rotate them such that the excess bandwidth demand across all angles is minimized. We define the compatibility score as $score =1 - $average($Excess(demand_{\alpha}))$, where $Excess$ is the excess bandwidth demand of all jobs at a particular angle $\alpha$ (Equation~\ref{eq:opt_def_loss}). To make the score a unitless metric, we divide the average excess bandwidth by the link capacity $C^{l}$. Note that when the excess bandwidth demand is zero, the compatibility score is 1 (i.e., 100\% compatible). However, when there are many jobs with large excess bandwidth demands, it is possible for the score to become negative indicating a highly incompatible combination. The optimization objective is to maximize this compatibility score, and the output of the optimization is a rotation angle $\Delta_j^l$ for each job.

\begin{table}[tbp]
\scriptsize
\centering
\def\arraystretch{1}
\begin{tabular}{|p{0.6cm}|p{1.8cm}|p{4.9cm}|}
\hline
\multirow{7}{1.22cm}{Input}
& $J^l = \{j\}$ & Set of ML jobs $j \in J^l$ competing on link $l$. \\
& $\{$unified\_circle$_j\}$ & Set of unified circles for  $\forall j \in J$. Each circle is a data structure that contains the angles and bandwidth demand of \on or \off phases.\\
& $bw\_circle_{j}(\alpha)$ & Bandwidth demand at angle $\alpha$ on unified\_circle$_j$\\
& $r_{j}$ & Number of iterations of $j$ in its unified\_circle$_j$.\\
& $A = \{\alpha\}$ & Set of discrete angles $\alpha \in$ [0,2$\pi$]. $|A|$ denotes the number of discrete angles.\\
& $C^{l}$ & Total link capacity of link $l$.\\
\hline
\multirow{2}{1cm}{Output} & \multirow{1}{1cm}{$demand_{\alpha}$} & Total bandwidth demand at angle $\alpha$ when considering the demand of all jobs $j \in J$.\\ 
& {$\Delta_j^l$} & Rotation angle of $j \in J$ on link $l$, in radians. \\
& $score$ & Compatibility score of jobs sharing link $l$.\\
\hline
\end{tabular}
\begin{align} 
\textbf{Auxiliary definitions:} & & \nonumber \\
Excess(demand_{\alpha}) = & \begin{cases} demand_{\alpha} - C^{l} & if demand_{\alpha} > C^{l} \\
0 & otherwise
\end{cases} & & \label{eq:opt_def_loss} \\
\textbf{Maximize:}  & \  score = 1 - \frac{ \sum_{\alpha}Excess(demand_{\alpha})}{|A|C} & & \label{eq:opt_objective} \\
\textbf{Subject to:} & & & \nonumber \\
\forall \alpha : & \sum_{j} bw\_circle_j(\alpha - \Delta_j^l) \leq demand_{\alpha} & & \label{eq:opt_constraint} \\
\forall \Delta_j^l : & \  0 \leq \Delta_j^l \leq \frac{2 \pi}{r_{j}} & & \label{eq:delta_bounds}
\end{align}
\vspace{-0.4cm}
\caption{\name's optimization formulation.}
\label{tab:formulation_phase_2}
\end{table}

\para{Optimization constraints.} 
Equation~\ref{eq:opt_constraint} computes the sum of the bandwidth demands across all the jobs sharing link $l$ at a particular angle $\alpha$ on their unified circles, rotated by angle ${\Delta_{j}^{l}}$. We bound this value by the output parameter $demand_{\alpha}$. Equation~\ref{eq:delta_bounds} bounds the rotation angle $\Delta_j^l$ between 0 and $\frac{2\pi}{r_j}$ because the unified\_circle$_j$ contains $r_j$ iterations of job $j$. Hence, setting an upper limit of $\frac{2\pi}{r_j}$ ensures that the rotation angle is in the first iteration and eliminates duplicate solutions.

\section{Augmenting ML schedulers with \namebf}
\label{sec:design}

This section describes how \name extends its link-level geometric abstraction to the entire cluster.

\subsection{\namebf \affinity Graph}
\label{sec:dependency_graph}

\para{Translating angular rotations to time-shifts.} Consider a set of jobs $j \in J^l$ sharing link $l$. Using the formulation in Table~\ref{tab:formulation_phase_2}, \name computes a rotation angle $\Delta_j^l$ for  $\forall j \in J^l$ such that the compatibility score is maximized. Each $\Delta_j^l$ corresponds to a time-shift $t_j^l$ to delay the start time of $j$ to maximize its compatibility with all other jobs in $J^l$. Given that the perimeter of the unified circle, $p^l$, is the LCM of the iteration times of all jobs $j \in J^l$, \name computes these time-shifts by multiplying the normalized rotation angle with $p^l$.
Formally:

\begin{equation}
\forall j \in J^l, t_j^l = (\frac{\Delta_j^l}{2\pi} \times p^l) \mod \textit{iter\_time}_j
\label{eq:time_shift}
\end{equation}

\para{Challenge: ensuring a unique time-shift for each job.} In a large-scale cluster, jobs are likely to traverse multiple links, and they may compete with different jobs on different links. Consider a case depicted in Figure~\ref{fig:challenges} where job $j_1$ competes with job $j_2$ on link $l_1$, and $j_2$ competes with job $j_3$ on link $l_2$. Theoretically, it is possible to migrate the jobs to pack workers of the same job under the same rack to avoid sharing the links altogether, but our experiments show that today's ML scheduling systems frequently end up with fragmented placements because of the dynamic nature of their scheduling decisions and job arrival patterns. In fact, no scheduler guarantees it can maintain perfect placement throughout time without continuously migrating jobs to defragment the cluster. For the case depicted in Figure~\ref{fig:challenges}, computing the time-shifts of $j_2$ using Equation~\ref{eq:time_shift} would result in two time-shift values $t_{j_2}^{l_1}$ and $t_{j_2}^{l_2}$. Given the interdependence between all servers participating in a training job, \name must find a \textit{unique} time-shift value for each job across links without compromising the compatibility on any link.

\begin{figure}[t]
\centering
\includegraphics[width=\columnwidth]{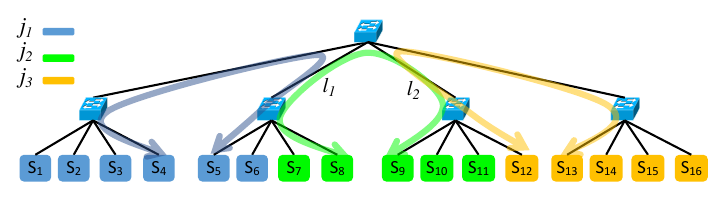}
\caption{Example illustrating a cluster-scale compatibility challenge: \name must ensure a unique time-shift for $j_2$.}
\label{fig:challenges}
\end{figure}

\para{Simple approach.} A potential approach to address the above challenge is to simply break the tie by choosing one of the $t_j^l$ values at random. But this approach cancels out the benefits of compatibility because it does not respect the carefully computed time-shifts for different links.

\begin{figure}[t]
\centering
\includegraphics[width=0.75\columnwidth]{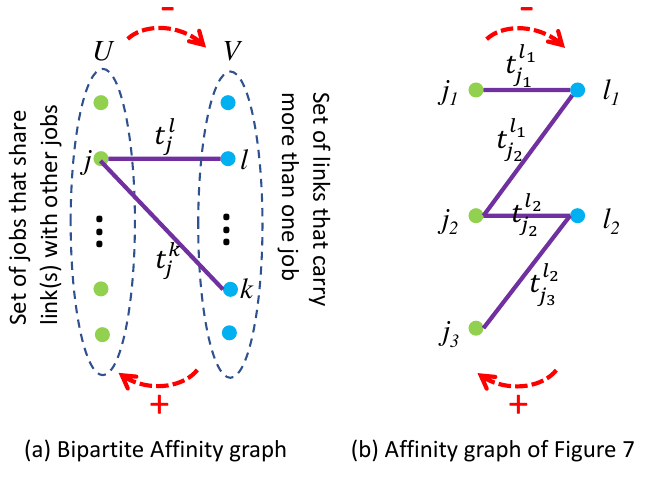}
\caption{\name's \affinity graph. Traversing left to right incurs a negative sign on the weight of edges and vice versa.}
\vspace{-0.3cm}
\label{fig:affinity}
\end{figure}

\para{Complex approach.} Another potential approach is to expand the footprint of our geometric abstraction from link-level to cluster-level. This approach requires expanding the optimization formulation in Table~\ref{tab:formulation_phase_2} to include all jobs that share their paths with any other jobs in the cluster and to encode a unique $\Delta_j$ in the constraints. This approach is not scalable because it requires expanding the perimeter of the unified circle to become the LCM of the iteration times of a large number of jobs in the cluster.  
Thus, finding a unique rotation angle for each job requires adding an  exponential number of constraints to the optimization formulation which  increases the complexity and overhead of the formulation dramatically.

\begin{algorithm}[t]
\small
\begin{algorithmic}[1]
\Procedure{BFSAffinityGraph}{}
    \IOComment{\textbf{Input} $Graph~G=(U, V, E)$: \name's \affinity graph}
    \IOComment{\textbf{Output} time\_shifts$_G$: Time-shifts of jobs in G}
    \State time\_shifts$_G$ = \{\}
    \ForAll {connected subgraphs $H \in G$, $H=(U_H, V_H, E_H)$} \label{line:subgraph}
        \State time\_shifts$_H$ = \{\}
        \textcolor{cyan}{\OldStatex \hspace{0.8cm} \(\triangleright\) \textit{BFS traversal}}
        \State Mark all vertices $u \in U_H$ as not-visited
            \State u = randomly\_select\_vertex($U_{H}$) \label{line:bfs_start}
            \State $t_u = 0$ and mark $u$ as visited
            \textcolor{cyan}{\OldStatex \hspace{0.8cm} \(\triangleright\) \textit{Only enqueue vertices from $U$ (jobs)}}            
            \State Q.enqueue($u$) \label{line:enqueue_1}
            \While{Q is not empty}
                \State $j =$ Q.dequeue()
                \textcolor{cyan}{\OldStatex \hspace{1.3cm} \(\triangleright\) \textit{Find the corresponding links and jobs}}
                \ForAll{neighbors $l$ of $j$}
                    \ForAll{neighbors $k$ of $l$}
                            \If {$k$ is not visited}
                            \State Q.enqueue($k$) and mark $k$ as visited \label{line:enqueue_2}
                            \textcolor{cyan}{\OldStatex \hspace{2.8cm} \(\triangleright\) \textit{Find the edge from $U$ to $V$}}
                            \State $e_1 = E_H(j,l)$ \label{line:e1}
                            \textcolor{cyan}{\OldStatex \hspace{2.8cm} \(\triangleright\) \textit{Find the edges from $V$ to $U$}}
                            \State $e_2 = E_H(l,k)$
                            \textcolor{cyan}{\OldStatex \hspace{2.8cm} \(\triangleright\) \textit{Compute the final time-shift}}
                            \State $t_k = (t_j - w_{e_1} + w_{e_2}) \%$ $iter\_time_k$
                            \State time\_shifts$_H[k] = t_k$
                            \label{line:traverse_formula}
                        \EndIf
                    \EndFor
                \EndFor
            \EndWhile
        \State time\_shifts$_G$ = time\_shifts$_G$ $\cup$ time\_shifts$_H$ 
    \EndFor
    \State \Return time\_shifts$_G$
\EndProcedure
\end{algorithmic}
\caption{Traversing the \affinity graph \label{alg:affinity}}
\end{algorithm}

\para{\namebf's approach.} \name introduces a bipartite \affinity graph $G=(U,V,E)$, where $U$ and $V$ are two sets of vertices, and $E$ denotes the edge set between $U$ and $V$, shown in Figure~\ref{fig:affinity}(a). Each vertex $u \in U$ represents a job that is sharing its path with other jobs somewhere in the network. Each vertex $v \in V$ represents a link that carries more than one job. An undirected edge $e = (j, l) \in E$ exists between a job $j \in U$ and a link $l \in V$ if $j$ traverses $l$. The weight of edge $e=(j, l) \in E$ is the time-shift of job $j$ on link $l$; i.e., $w_e = t_j^l$.

\para{Traversing the \affinity graph.} \name uses a graph traversal algorithm to find unique time-shifts $t_j$ for all jobs $j \in J$ while maintaining compatibility on all links. To consolidate $t_j^l$ values for each job $j$ and link $l$ into a unique $t_j$ value, \name first randomly selects one of the jobs in the \affinity graph as the reference point with $t_j = 0$ and then traverses the graph to compute unique time-shifts for all others. Algorithm~\ref{alg:affinity} describes the pseudocode of our graph traversal. In the general case, the \affinity graph is not necessarily a connected graph, hence, the algorithm traverses each connected subgraph separately (line~\ref{line:subgraph}). The traversal algorithm extends the Breadth First Search (BFS) algorithm in two ways. First, only vertices in $U$ are added to the BFS queue ($Q$) because the time-shifts correspond to jobs,  not links (lines~\ref{line:bfs_start}- \ref{line:enqueue_2}). Second, traversing from jobs ($j \in U$) to links ($l \in V$) incurs a negative sign on the $t_j^l$ weight on edge $e=(j,l)$, whereas traversing the reverse direction incurs a positive sign (lines~\ref{line:e1}-\ref{line:traverse_formula}). As soon as the vertex corresponding to job $j$ is visited, its  unique time-shift is determined by the algorithm (line~\ref{line:traverse_formula}). 

\newtheorem{theorem}{Theorem}
\begin{theorem}[Correctness and Uniqueness Guarantee]
\label{thm:uniqueness}
Given a cluster with $J$ jobs and a loop-free \affinity graph, $G=(U,V,E)$, Algorithm~\ref{alg:affinity} guarantees both correct and unique time-shifts $t_j$ for all jobs $j \in J$.
\end{theorem}

\begin{proof}
The key insight behind this Theorem is that our graph traversal maintains the same \textit{relative} time-shift for all job pairs in the \affinity graph. The full proof uses induction and is provided in Appendix~\ref{appendix:proof} along with an example corresponding to the \affinity graph in Figure~\ref{fig:affinity}(b).
\end{proof}

\begin{figure}[t]
    \centering
    \includegraphics[width=0.8\columnwidth]{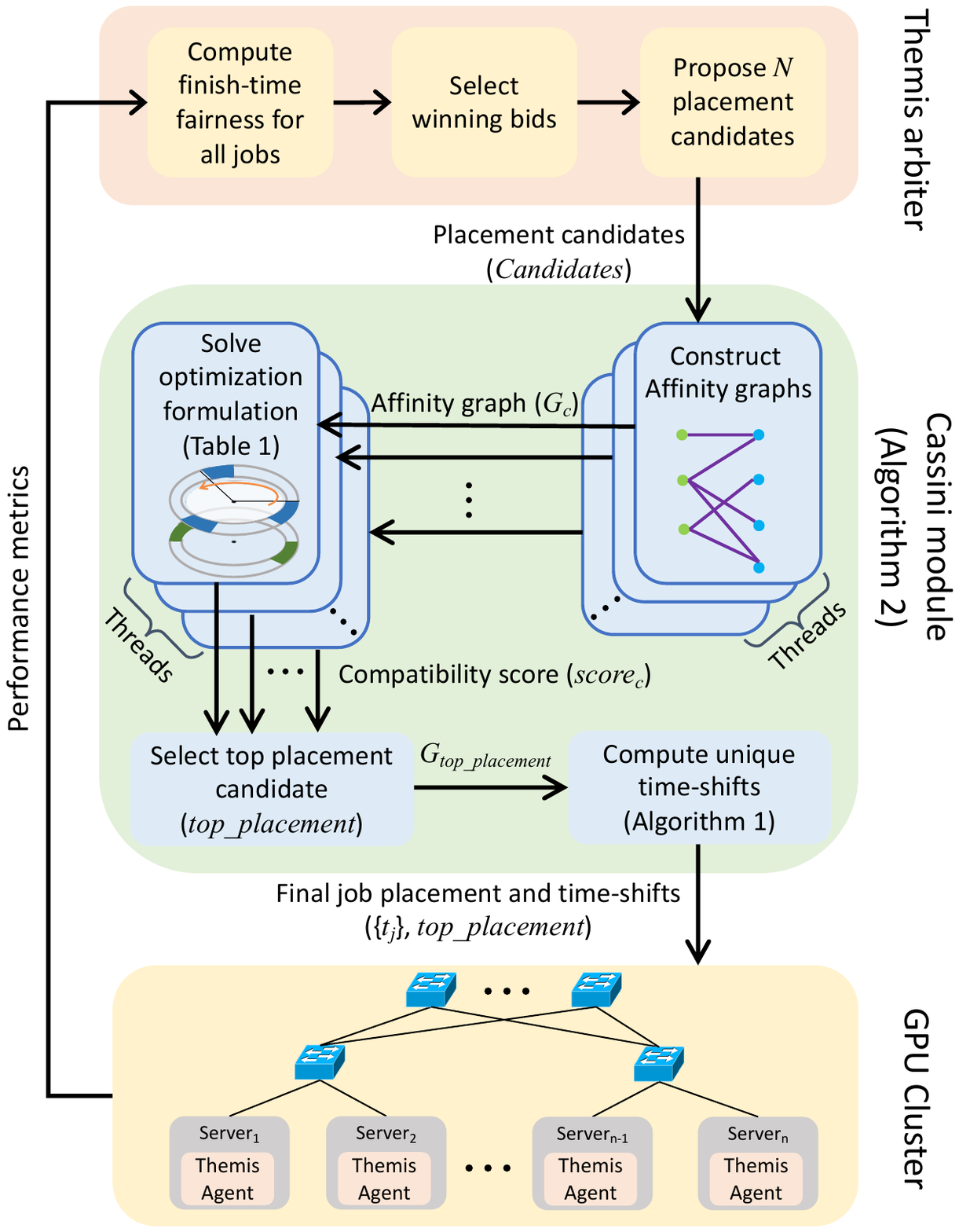}
    \caption{Using \name to augment Themis~\cite{themis}.}
    \label{fig:system}
\end{figure}

\subsection{Putting It All Together}
\label{sec:putting_togerher}

This section uses Themis~\cite{themis} as a running example of a scheduler augmented by \name. 

\para{Overview of Themis.} Themis uses a fairness metric, called finish-time fairness, to achieve long-term fairness across the entire cluster by periodically updating the placement of jobs. To achieve fairness, workers in Themis lease resources and go through periodic auction epochs to help jobs that are farthest in terms of their fairness metric bid for more resources. Themis's central arbiter determines the global winning bids to maximize the aggregate improvement in the finish-time fair metrics across all bidding jobs. To capture network overheads, Themis uses a slowdown penalty based on whether the workers are under the same rack or across racks.

\para{Augmenting Themis with \namebf.} Figure~\ref{fig:system} shows how \name augments Themis. First, \name modifies Themis's arbiter to return a set of potential placement candidates instead of a single placement. Then, \name selects the top placement candidate based on its compatibility metric and computes unique time-shifts for jobs that share the network. \name transfers the time-shifts to Themis's agent running on servers. Finally, Themis's agent applies the time-shifts at the start of the epoch. Note that \name respects the hyper-parameters, such as batch size or the number of workers, decided by Themis (or other schedulers that \name is augmenting).
Next, we describe each step in detail.

\setlength{\textfloatsep}{15pt}
\begin{algorithm}[t]
\small
\begin{algorithmic}[1]
\Procedure{CassiniModule}{}
    \IOComment{\textbf{Input} $Jobs$: Array of active training jobs in the cluster}
    \IOComment{\textbf{Input} $Links$: Array of all links in the cluster}
    \IOComment{\textbf{Input} $Candidates$: Array of candidate placements for jobs}
    \IOComment{\textbf{Output} $top\_placement$, $\{t_j\}$: Top placement and time-shifts}
    \For {$c \in Candidates$}  \textcolor{gray}{\hspace{0.05cm} \(\triangleright\) \textit{(Loop is executed with threads)}}
        \textcolor{cyan}{\OldStatex \hspace{0.8cm} \(\triangleright\) \textit{Construct \name's \affinity graph corresponding}}
        \textcolor{cyan}{\OldStatex \hspace{0.8cm} \textit{ to this placement (\S\ref{sec:dependency_graph})}}
        \State $G_c$ = $(U_c, V_c, E_c)$ \label{line:affinity_begin}
        \ForAll {$j \in Jobs$, $l \in Links$}
            \If{ $j$ shares links with other jobs}
                \State $U_c = U_c \cup j$
                \If {$l$ carries more than one job}
                    \State $V_c = V_c \cup l$
                        \If {$j$ is traversing $l$}
                            \State $e=$ new Edge between $\{(j, l)\}$
                            \State $E = E \cup e$
                            \State $w_e = 0$ \label{line:affinity_end}
                        \EndIf
                \EndIf
            \EndIf
        \EndFor
        \textcolor{cyan}{\OldStatex \hspace{0.8cm} \(\triangleright\) \textit{Discard this candidate if \affinity graph has a loop}}
        \If{there is a loop in $G_c$}
            \State $Candidates.remove(c)$
            \State continue \label{line:loop-detection}
        \EndIf
        \State $score_c$ = \{\}
        \For{$l \in V_c$} \textcolor{gray}{\hspace{0.05cm} \(\triangleright\) \textit{(Executed with threads)}} \label{line:compute_compatbility_begin} 
            \textcolor{cyan}{\OldStatex \hspace{1.2cm} \ \ \ \(\triangleright\) \textit{List of jobs traversing link $l$}}
            \State $J^l = \{\}$
            \ForAll {neighbors $j$ of $l$}
                \State $J^l = J^l \cup j$
            \EndFor
            \textcolor{cyan}{\OldStatex \hspace{1.2cm} \ \ \ \(\triangleright\) \textit{Solve \name optimization (Table~\ref{tab:formulation_phase_2})}}
            \State $score_l$ = \Call{CassiniOptimization}{$J^l$}
            \State $score_c = score_c \cup score_l$ \label{line:compute_compatbility_end} 
        \EndFor
        \textcolor{cyan}{\OldStatex \hspace{0.6cm} \ \ \ \(\triangleright\) \textit{Set the compatibility score of candidate $c$}}        
        \State $c.score = score_c$
    \EndFor
    \textcolor{cyan}{\OldStatex \hspace{0.4cm} \(\triangleright\) \textit{Sort placements based on compatibility metric}} 
    \State \Call{SortCandidates}{$Candidates$, ``$Decreasing$"} \label{line:top_candidate_begin}
    \State $top\_placement = Candidate[0]$ \label{line:top_candidate_end}
    \textcolor{cyan}{\OldStatex \hspace{0.4cm} \(\triangleright\) \textit{Find unique time-shifts (Algorithm~\ref{alg:affinity})}} 
    \State $\{t_j\}$ = \Call{BFSAffinityGraph}{$G_{top\_placement}$} \label{line:time_shifts_begin}
    \State \Return $\{t_j\}, top\_placement$ \label{line:return_time_shifts}
\EndProcedure
\end{algorithmic}
\caption{\name Module's Pluggable Algorithm \label{alg:cassini_compatibility}}
\end{algorithm}

\para{Step 1. Discover placement candidates.} In this step, \name decouples the process of finding the number of workers for each job to improve finish-time fairness from the exact worker placement in the cluster. To do so, instead of returning the precise job placements at the end of the auction phase, we configure Themis to return up to $N$ \textit{candidate placements}. These candidate placements all achieve the same finish-time fairness, but their worker placements are different. For instance, consider a case where jobs $j_1$ and $j_2$ each place a bid on two additional workers, and they both win, while job $k_1$ is losing one worker, and job $k_2$ is losing three. In this case, there are two ways to distribute workers: $(i)$ $k_1$ and $k_2$ each give up one worker to $j_1$, and $k_2$ gives two workers to $j_2$; or $(ii)$ $k_1$ and $k_2$ each give up one worker to $j_2$, and $k_2$ gives two workers to $j_1$. Both options are candidate placements. Moreover, selecting which workers in $k_1$ and $k_2$ should be reassigned creates another set of candidate placements. \name collects these candidate placements and feeds them as input to the next step. This process requires changing only $\approx$300 lines of code in Themis.

\para{Step 2. Find unique time-shifts.} This step is listed in Algorithm~\ref{alg:cassini_compatibility} and includes \name's key contributions. \name first constructs an \affinity graph $G_c$ for each placement candidate $c \in Candidates$ (lines~\ref{line:affinity_begin}-\ref{line:affinity_end}). Following Theorem~\ref{thm:uniqueness}, to ensure correctness, we discard placement candidates with loop(s) in any of their \affinity subgraphs (line~\ref{line:loop-detection}). Then, \name constructs the unified circles for each job and solves the optimization formulation in Table~\ref{tab:formulation_phase_2} for all links in $G_c$ to obtain the compatibility metric for each link in $V_c$ (lines~\ref{line:compute_compatbility_begin}-\ref{line:compute_compatbility_end}). Given that the placement candidates are independent of each other, our implementation uses multiple threads to parallelize this computation. Once the compatibility score of all candidate placements is determined, \name sorts each placement candidate based on the average compatibility score of its member links to find the top placement candidate $top\_placement \in Candidates$  (lines~\ref{line:top_candidate_begin}-\ref{line:top_candidate_end}).\footnote{Instead of averaging, tail or other metrics may also be used.} Then, it executes Algorithm~\ref{alg:affinity} on $top\_placement$'s \affinity graph $G_{top\_placement}$ to obtain unique time-shifts $\{t_j\}, \forall j \in V_{top\_placement}$ for jobs that share links with other jobs in this placement (line~\ref{line:time_shifts_begin}). Finally, $top\_placement$ and its corresponding time-shifts are transferred to Themis's agent running on the servers (line~\ref{line:return_time_shifts}).

\para{Step 3. Apply time-shifts.}  When a time-shift $t_j$ is received by the Themis agent running job $j$, it delays the start of the next immediate training iteration by $t_j$. However, even though the workers of the same job apply a unique time-shift, the time-shift values can drift due to noise, stragglers, and other unpredictable events. \name updates the agent on each server to measure the drift and adjust the time-shifts. Our evaluations show that time-shift adjustments are rare (\S\ref{sec:testbed_time_shift_adjustment}).

\section{Evaluations}
\label{sec:testbed}

We evaluate \name on a 24-server cluster and compare its performance against other state-of-the-art ML schedulers. First, we describe our evaluation methodology and setup (\S\ref{sec:eval_setup}). Then, we compare \name's performance gains with respect to the state-of-art ML schedulers for a mix of data and model parallel DNN training jobs (\S\ref{sec:testbed_performance_gains}). Next, we evaluate the impact of data parallelism (\S\ref{sec:testbed_congestion_data_parallel}), model parallelism
(\S\ref{sec:model_parallel}), partial compatibility (\S\ref{sec:testbed_timeshift_gains}), and having multiple GPU per server on \name's performance (\S\ref{sec:multiple_gpus}). Finally, we evaluate the frequency of time-shift adjustments and \name's overhead (\S\ref{sec:testbed_time_shift_adjustment}). \name's code will be publicly available.

\begin{figure*}[t]
    \begin{minipage}{0.65\textwidth}
    \includegraphics[width=\textwidth]{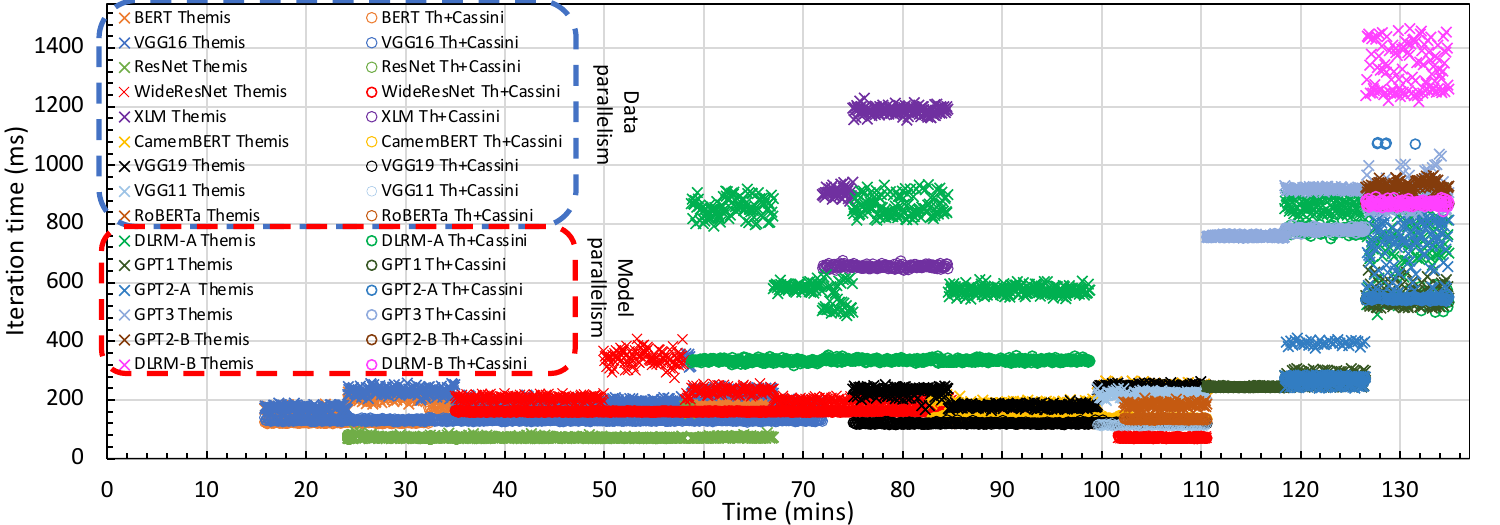}
    \caption{[Poisson trace] Time series of data parallel and model parallel jobs.}
    \label{fig:scatter_epoch}
    \end{minipage}
    \hspace{0.2cm}
    \begin{minipage}{0.3\textwidth}
    \includegraphics[width=\textwidth]{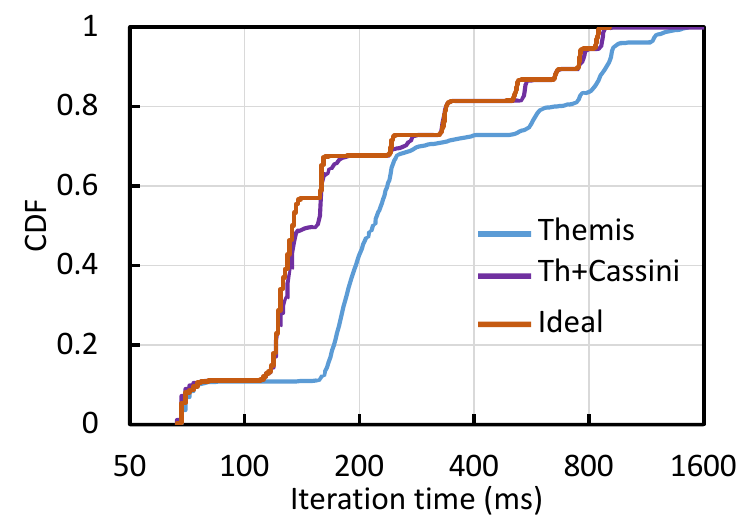}
    \caption{[Poisson trace] CDF of the iteration times.}
    \label{fig:scatter_epoch_cdf}
    \end{minipage}
\end{figure*}

\subsection{Methodology And Setup}
\label{sec:eval_setup}

\para{Setup.} We build a prototype to demonstrate the gains of \name in real-world settings. Our prototype includes 24 ASUS ESC4000A-E10 servers each with one A100 Nvidia GPU~\cite{a100} (40~GB of HBM2 memory) and one 50~Gbps Mellanox ConnectX5 NIC. We use RoCEv2 for communication and enable DCB~\cite{dcb} and PFC on these interfaces to support a lossless fabric for RDMA. The servers run Ubuntu 18.04~LTS. We use PyTorch~\cite{pytorch} version 1.8.0, CUDA version 11.1, and NCCL version  2.11.4 in our training framework.

\para{Topology.} We use a Tofino switch to construct the logical topology illustrated in Figure~\ref{fig:testbed_topology} (Appendix~\ref{app:testbed_topology}) with 13 logical switches and 48 bi-directional links and 2:1 over-subscription above the ToRs. We use flow table rules that match on <input port, destination MAC> to forward packets to the correct output port and physical loopback cables for switch-to-switch links. We use the default RDMA-based DCQCN congestion control algorithm~\cite{dcqcn}. ECN is enabled through WRED with min and max thresholds set to 1000 and 2000 cells. The PFC skid buffer threshold of each virtual switch is 4000 cells.

\para{DNN Workloads.} We experiment with 13 popular DNN models: VGG11~\cite{vgg11}, VGG16~\cite{vgg16}, VGG19~\cite{vgg19}, ResNet50~\cite{resnet}, WideResNet101~\cite{wideresnet}, BERT~\cite{bert}, RoBERTa~\cite{roberta}, XLM~\cite{xlm}, CamemBERT~\cite{camembert}, GPT-1~\cite{gpt_1}, GPT-2~\cite{gpt_2}, GPT-3~\cite{gpt_3}, and DLRM~\cite{dlrm}. All Models have an equal probability of occurrence and the training duration time is randomly selected between 200 - 1,000 iterations. Table~\ref{tab:model_parameters} (Appendix~\ref{app:dnn_models_details}) provides details about model configurations and batch sizes used in this paper.

\para{Parallelization strategy.} We use data parallelism to train VGG, ResNet, and BERT family of models using Pytorch's DistributedDataParallel framework~\cite{pytorchDDP}. This framework distributes the dataset across GPUs and uses RingAllreduce to update the gradients during each training iteration. We train DLRM and GPT family of models using a hybrid of data/model parallelism. Following prior work~\cite{topoOpt}, we use Meta's opensource codebase for training DLRM~\cite{dlrm} where the embedding tables are partitioned across GPUs, while the rest of the model is replicated on all GPUs. Finally, we use Microsoft's DeepSpeed tool~\cite{deepspeed_gpt} to partition the GPT models across GPUs using hybrid data/model parallelism.

\para{Traces.} Following prior work~\cite{synergy_osdi22, muri, themis, pollux}, we use three sets of traces in our evaluations: ($i$) \textit{Poisson trace}: we use a Poisson distribution for job arrivals where the job arrival time is determined by the load parameter defined as the average fraction of GPUs that are serving active jobs in the cluster. We vary the load between 80\% and 100\%; ($ii$) \textit{dynamic trace}: where a set of DNN training jobs are present in the cluster and a new set of jobs arrive; ($iii$) \textit{snapshot trace}: we take several snapshots of the cluster where all jobs are present at the start of the experiment. 

\begin{figure*}[t]
    \includegraphics[width=\textwidth]{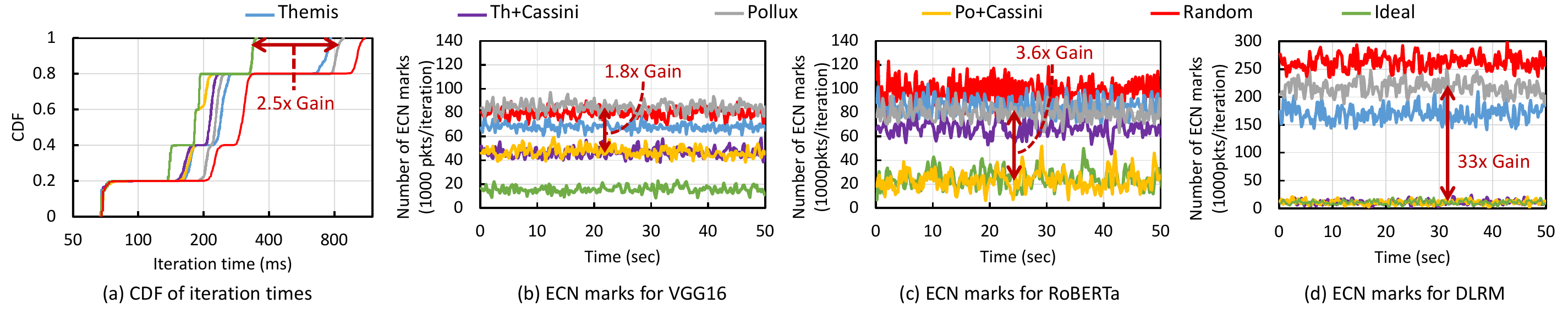}
    \caption{[Dynamic trace] CDF of training iteration times and the number of ECN marked packets per iteration.}
    \vspace{-0.2cm}
    \label{fig:cdf_congestion_testbed_dp}
\end{figure*}

We implement the following schemes in our testbed.

\begin{itemize}[align=left, leftmargin=0pt, labelindent=0pt, listparindent=\parindent, labelwidth=0pt, itemindent=!]
\itemsep0em 

\item \textbf{Themis.} We use the default Themis~\cite{themis} scheduler as one of our baselines. The bidding period (epoch) is set to 10 mins. Jobs participate in an auction where they send bid values for different GPU allocations. An assignment of GPU servers is valid until the period of lease time. When the lease time expires, the job gives up the server and a new auction is conducted for all the released servers. When a job arrives, its initial number of requested workers is randomly selected between 1 to 12 GPUs. As the experiment progresses, the number of workers is automatically tuned based on Themis's finish-time-fairness metric.

\item \textbf{Th+\namebf.} Themis augmented with \name as described in Section~\ref{sec:putting_togerher}. In particular, this scheduler takes up to 10 placement candidates from Themis, constructs geometric circles and \affinity graphs for each placement to capture the cluster-level compatibility, solves our optimization formulation to find time-shifts for jobs that are competing on bandwidth, selects the top placement candidate based on compatibility ranks, and finally computes a unique time-shift for jobs. The unique time-shifts and final placement are given to the Themis agent running on GPUs. Unless otherwise stated, we use 5$^{\circ}$ as the angle discretization precision (Table~\ref{tab:formulation_phase_2}) to compute the time-shifts.

\item \textbf{Pollux.} We use Pollux as a second baseline~\cite{pollux}. Pollux considers the system throughput and statistical efficiency to maximize cluster-wide training performance. Pollux periodically queries jobs and reassigns GPUs to maximize the overall goodput of the cluster. Pollux also models migration costs and avoids frequent job migrations.

\item \textbf{Po+\namebf.} We augment Pollux with \name using a similar approach described in Section~\ref{sec:putting_togerher} except that Pollux uses overall goodput, instead of finish-time-fairness, to adjust hyper-parameters during scheduling epochs. Hence, the number of workers assigned to each job does not always agree with Themis. To make an apples-with-apples comparison, all \name-related parameters in Po+\name and Th+\name are the same.

\item \textbf{Ideal.} An ideal scheduler that runs each training job on a dedicated cluster. This scheduler incurs no congestion since the entire cluster is dedicated to one job and there is no need to take job compatibility into account.

\item \textbf{Random.} A random placement scheduler that places workers for each job randomly. This scheduler has the highest network overhead since it does not take locality or compatibility into account.

\end{itemize}

\para{Profiling DNN models.} Similar to Themis and Pollux, we profile each DNN using Pytorch and Infiniband port counters. Our profiling script executes a few iterations of each job to measure iteration times and collect link utilization patterns for various batch sizes and numbers of workers. Fine-grained link utilization data from the port counters enables \name to build the geometric circles and the corresponding bandwidth demands for our optimization  ($bw\_circle_{j}(\alpha)$ in Table~\ref{tab:formulation_phase_2}).

\begin{figure*}[t]
    \includegraphics[width=\textwidth]{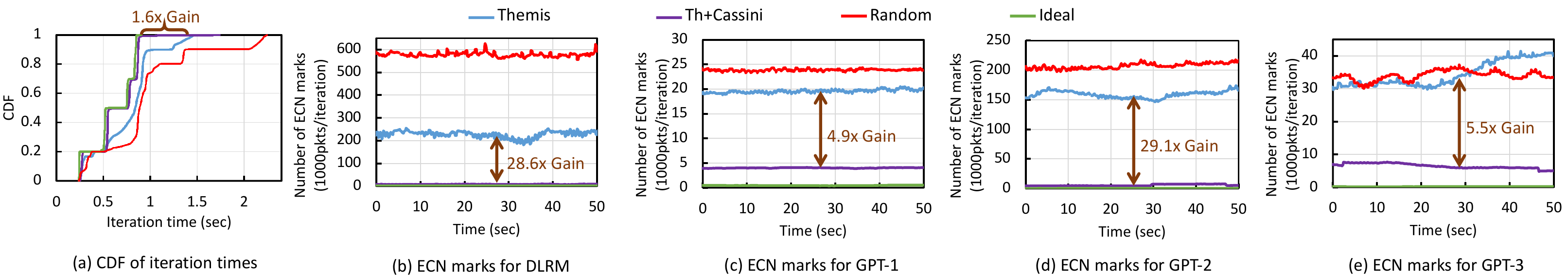}
    \caption{[Dynamic trace, model parallelism] CDF of training iteration times and the number of ECN marked packets.}
    \vspace{-0.2cm}
    \label{fig:cdf_congestion_testbed_mp}
\end{figure*}

\subsection{Performance Gains}
\label{sec:testbed_performance_gains}

We evaluate \name's performance gains using job arrivals and departures from our Poisson trace.  Figure~\ref{fig:scatter_epoch} plots the time series of events in the cluster for Themis and Th+\name. In this experiment, we use a combination of data parallel and model parallel training jobs. Placement changes are triggered by job arrivals, departures, and when the lease time of any of the servers expires. Given the dynamic nature of the trace, the servers are occupied gradually and their lease times are not synchronized.  For instance, at time $t=72$~mins a data parallel training job for the XLM~\cite{xlm} model arrives at the cluster and Themis places it such that one of the links is shared with WideResNet101~\cite{wideresnet} without the knowledge that XLM and WideResNet101 are not compatible jobs. 
In contrast, Th+\name improves the iteration time of XLM by placing it with compatible jobs.
Note that this trace contains different training instances of the same DNN models where they differ in their hyper-parameters and number of workers (details in Appendix~\ref{app:dnn_models_details}). For instance, at time $t=118$~min, a model parallel GPT-2~\cite{gpt_2} training job (labeled as GPT-2-A) arrives at the cluster and without considering the communication demands, Themis places this job such that it shares a link with another large GPT-3~\cite{gpt_3} model in the cluster. GPT-2-A and GPT-3 models are not compatible making both training jobs slow down. In contrast, Th+\name improves GPT-2-A's iteration time by placing it with a compatible GPT-1 model.

Figure~\ref{fig:scatter_epoch_cdf} plots the CDF of iteration times of all the data points in Figure~\ref{fig:scatter_epoch} and shows that compared to Themis, Th+\name improves the average and 99$^{th}$ percentile tail iteration times by 1.4$\times$ and 1.5$\times$,
respectively. We observe similar gains between Po+\name and Pollux.  
The figure also shows that Th+\name achieves similar performance as our Ideal benchmark. 

\subsection{\namebf Reduces Congestion}
\label{sec:testbed_congestion_data_parallel}

We next demonstrate the effectiveness of \name in reducing the congestion in the network.
We use our dynamic trace to trigger the arrival of DLRM and ResNet50 to the cluster while the cluster is busy running other jobs. Given the contrast between the network demand between these two models, this experiment serves as a stress test to highlight the importance of compatible job placement on network congestion. In this case, both Pollux and Themis end up placing DLRM on servers that share network links with other non-compatible jobs which hurts the iteration times. In contrast, Th+\name and Po+\name flip the placements of DLM and ResNet50 to achieve compatibility, thereby improving the training iteration times, as depicted in Figure~\ref{fig:cdf_congestion_testbed_dp}(a). Compared to Themis, Th+\name improves the average and 99$^{th}$ percentile tail iteration times by 1.5$\times$ and 2.2$\times$, respectively. Similarly, compared to Pollux, Po+\name improves the average and 99$^{th}$ percentile tail iteration times by 1.6$\times$ and 2.5$\times$, respectively. 

The gains in iteration times are a direct consequence of \name's ability to reduce network congestion. Figures~\ref{fig:cdf_congestion_testbed_dp}(b) to (d) show the number of ECN-marked packets per iteration for different models. The figure shows that Th+\name and Po+\name consistently maintain a lower number of ECN marks per iteration across the models. In particular, Figure~\ref{fig:cdf_congestion_testbed_dp}(d) shows that, on average, DLRM is experiencing 27$\times$ and 33$\times$ more ECN marks in Themis and Pollux, compared to their \name-augmented counterparts.

\begin{figure*}[t]
    \centering
    \includegraphics[width=\textwidth]{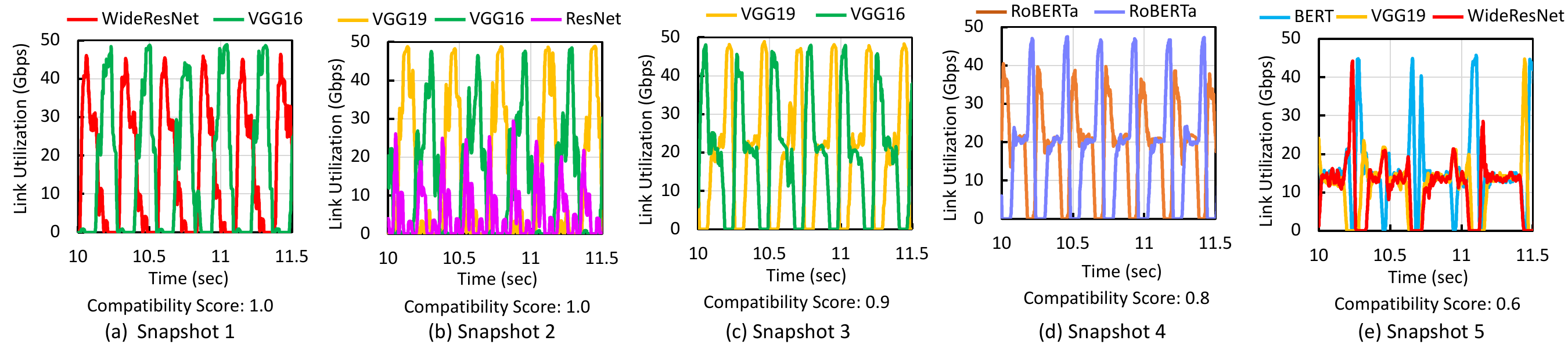}
    \caption{[Snapshot trace] Link utilization of compatible and partially compatible snapshots.}
    \label{fig:bw_share_compatible}
\end{figure*}

\subsection{Impact of Model Parallelism}
\label{sec:model_parallel}

To ensure \name's gains are not limited to data parallel jobs, we run a series of experiments when all jobs in the trace are using model parallelism. As shown in Section~\ref{sec:traffic_pattern}, model parallel jobs have several \on and \off phases in each iteration where the duration and bandwidth demand of each phase depends on the parallelization strategy and hyper-parameters. Similar to the data parallel case, we use \name's geometric abstraction to capture the duration and bandwidth demand of \on and \off phases of a series of model parallel jobs. We then use \name's optimization formulation and \affinity graph to compute the time-shifts for the jobs sharing the same network links. We use our dynamic trace to trigger the arrival of multiple GPT and DLRM models while the cluster is training other model parallel jobs.

\begin{table}
\scriptsize
\begin{tabular}{|p{0.4cm}|p{2.1cm}|p{1.1cm}|p{0.8cm}|p{0.6cm}|p{0.7cm}|}
\hline 
Snap-shot ID & Competing jobs (batch size) &  Th+\name & Themis & Comp-atibility score & time-shift (ms)  \\ \hline\hline
\rowcolor{LightCyan} 1 & WideResNet101~(800)  &  138 ms   & 205 ms & 1.0 & 0 ms \\
\rowcolor{LightCyan}        & VGG16~~(1400)  & 148 ms  & 199 ms  & &150 ms\\
\hline
\rowcolor{LightCyan} 2 & VGG19~~(1400)  &  168 ms & 220 ms & 1.0 & 0 ms  \\
\rowcolor{LightCyan}       & VGG16~~(1700)  &  163 ms  & 220 ms & & 158 ms\\
\rowcolor{LightCyan}       & RESNET50~~(1600) &  59 ms  & 55 ms & & 46 ms \\
\hline
\rowcolor{LightCyan} 3 & VGG19~~(1024)  &  166 ms & 176 ms & 0.9 & 0 ms \\
\rowcolor{LightCyan}       & VGG16~~(1200)  &  168 ms & 177 ms  & &  100 ms\\
\hline
\rowcolor{LightCyan}       4 & RoBERTa~~(12) & 164 ms  & 210 ms & 0.8 & 0 ms \\
 \rowcolor{LightCyan}        & RoBERTa~~(12) & 180 ms  & 208 ms & & 60 ms  \\
\hline
5 & BERT~~(8)      & 209 ms & 213 ms & 0.6 & 0 ms \\
       & VGG19~~(1400)  &  294 ms  & 292 ms & & 42 ms\\
       & WideResNet101~(800) &  265 ms & 266 ms & & 191 ms\\
\hline
\end{tabular}
\captionof{table}{[Snapshot trace] Compatibility score of DNN jobs.}
\label{table:compatibility}
\end{table}

Figure~\ref{fig:cdf_congestion_testbed_mp}(a) shows the CDF of training iteration times. We find that similar to the data parallel case, Themis ends up placing non-compatible jobs, such as <GPT-3 and GPT-2> or <GPT-1 and DLRM>, on the same network link, whereas Th+\name places compatible jobs, such as <GPT-1 and GPT-2> or <GPT-3 and DLRM>, on the same network links. Consequently, Th+\name improves the average and 99$^{th}$ percentile tail iteration times by 1.2$\times$ and 1.6$\times$, respectively. We observe similar gains between Pollux and Po+\name.

Figures~\ref{fig:cdf_congestion_testbed_mp}(b) to (e) depict the number of ECN marked packets per iteration for the models in this experiment. Depending on the status of congestion, different models experience different numbers of ECN-marked packets. For instance, Figure~\ref{fig:cdf_congestion_testbed_mp}(d) shows that compared to Themis, Th+\name reduces the average number of ECN marked packets by 29$\times$. 

\subsection{Impact of Partial Compatibility}
\label{sec:testbed_timeshift_gains}

An important consideration for practical deployment of \name is to evaluate the impact of placing \textit{partially} compatible jobs on the same link(s). Intuitively, the higher the compatibility score, the better interleaving is achieved. As the compatibility score reduces, the gains diminish. To evaluate the impact of partial compatibility, we take five snapshots of the cluster, as shown in Table~\ref{table:compatibility} and compute the compatibility scores and time-shift values from our optimization formulation (\S\ref{sec:geometric}) 
for each snapshot. We then measure the average communication time of each model under Themis and Th+\name. The table shows that when the compatibility score is 0.6, \name's gain compared to Themis starts to diminish. Note that \name avoids placing jobs with low compatibility score (e.g., snapshot 5) on the same link.

We demonstrate the reason behind diminishing returns by plotting the link utilization of each snapshot in Figure~\ref{fig:bw_share_compatible}. When the compatibility score is high, the opportunity for interleaving is large and jobs end up interleaving their network usage most of the time, as shown in Figures~\ref{fig:bw_share_compatible}(a)--(d). However, as the compatibility score is reduced, jobs are forced to share the link most of the time, as shown in Figure~\ref{fig:bw_share_compatible}(e). Additionally, Figure~\ref{fig:bw_share_compatible}(b) demonstrates a desirable feature of our optimization formulation where compatibility does not require \textit{strict} interleaving. In this snapshot, only VGG19 and VGG16 are interleaved and ResNet's communications overlap with both jobs because its network demand is not significant and can co-exist with the other jobs. 

\begin{figure}[t]
    \includegraphics[width=\columnwidth]{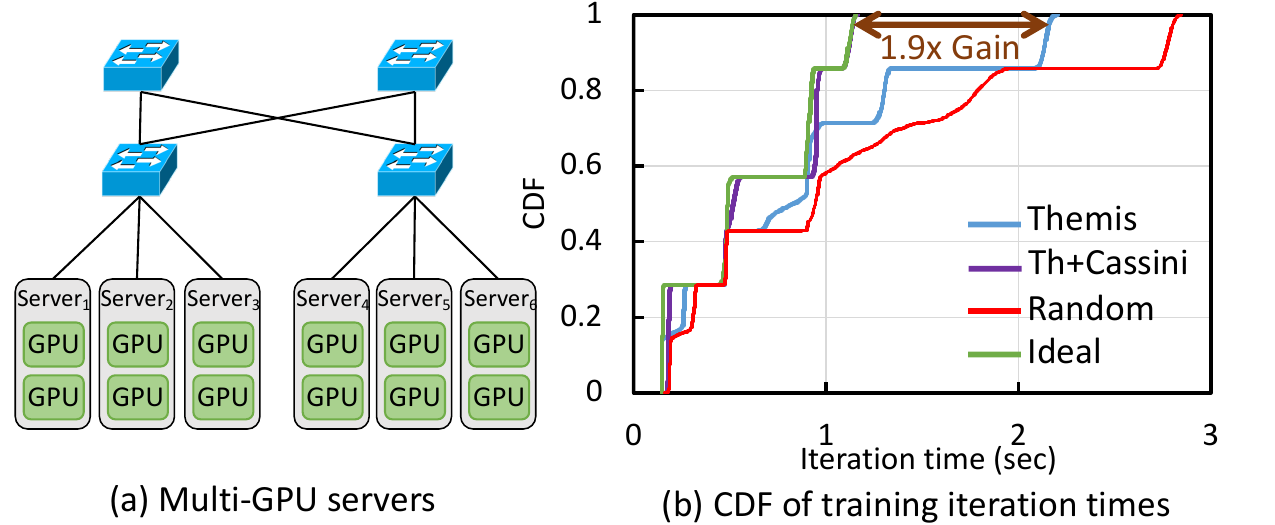}
    \caption{[Dynamic trace] Multi-GPU experiment.}
    \label{fig:multi_gpu}
\end{figure}

\subsection{Impact of Multiple GPUs per Server}
\label{sec:multiple_gpus}

Although having multiple GPUs per server enables allocating more GPUs within the same server to a job, today's large-scale training jobs require hundreds of workers~\cite{topoOpt, mudigere2021highperformance}, making it impossible to avoid network congestion entirely by relying on multi-GPU servers. In such cases, \name's gains are more pronounced for models that are distributed outside the boundary of a server.

We evaluate \name's gains with multi-GPU servers by removing GPUs from some of our single-GPU servers and adding them to other servers to compose servers with two GPUs. We create a topology with six servers, each with two GPUs, as shown in Figure~\ref{fig:multi_gpu}(a). We then use a mix of data parallel and model parallel jobs and generate a series of job arrivals using our dynamic trace.

Figure~\ref{fig:multi_gpu}(b) demonstrates that compared to Themis, Th+\name improves the average and 99$^{th}$ percentile tail iteration times by 1.4$\times$ and 1.9$\times$, respectively. These gains are achieved because some jobs require more than two GPUs to train. For instance, at a particular instant in our dynamic trace, the XLM and ResNet50 models each require three GPUs to train. With the arrival of a network-intensive model DLRM requesting three more GPUs, Themis decides to place DLRM such that it shares a server with a non-compatible model (XLM) making both jobs experience congestion. In contrast, Th+\name selects a placement where DLRM shares a link with a compatible model (ResNet50), thereby improving the training iteration times of both models. 

\subsection{Adjusting Time-Shifts and Overhead}
\label{sec:testbed_time_shift_adjustment}

\begin{figure}[t]
    \begin{minipage}{0.23\textwidth}
    \includegraphics[width=\textwidth]{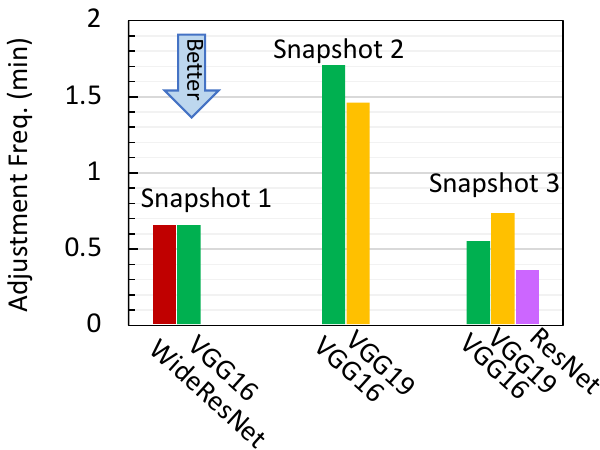}
    \caption{[Snapshot trace] The frequency of adjusting time-shifts for snapshots 1--3.}
    \label{fig:adjustment_bar}
    \end{minipage}
    \hspace{0.1cm}
    \begin{minipage}{0.23\textwidth}
    \includegraphics[width=\textwidth]{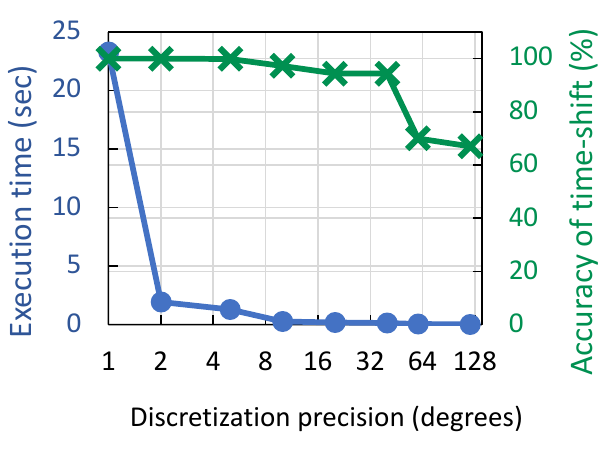}
    \caption{Impact of angle discretization on execution time and time-shift accuracy.}
   \vspace{-0.3cm}
    \label{fig:discretization_vs_perf}
    \end{minipage}
\end{figure}

To maintain \name's interleaving, workers must respect the time-shift values given to them through the scheduler. Given that our servers are not running perfectly in sync, we evaluate the frequency of automatic time-shift adjustments by the Themis (or Pollux) agents running on the servers. Note that respecting the time-shift is only required for compatible jobs. All other jobs in the cluster can send packets at any time. A worker triggers an adjustment when the start of the communication phase deviates by more than five percent of the ideal iteration time. Figure~\ref{fig:adjustment_bar} shows the average frequency of time-shift adjustments for snapshots 1,2, and 3. As shown, in all cases, the frequency is less than two adjustments per minute.

Finally, we evaluate the impact of angle discretization precision on \name's optimization formulation (Table~\ref{tab:formulation_phase_2}). Intuitively, the execution time of a coarse-grained discretization is fast but such formulation misses interleaving opportunities, thereby finding imprecise rotation angles. Given that \name's time-shifts are driven from rotation angles, a coarse-grained formulation leads to inaccurate time-shifts. On the other hand, having fine-grained precision leads to more accurate time-shifts at the expense of a longer execution time. Figure~\ref{fig:discretization_vs_perf} demonstrates this trend and shows that using a precision of 5$^{\circ}$ is the sweet spot for achieving 100\% accuracy for time-shifts while maintaining a low execution overhead.

\section{Related Work}
\label{sec:relatedwork}
\label{sec:prior_work}

Our work builds on several lines of related research.

\para{Compute scheduling approaches.} A large number of systems and techniques have focused on improving the performance of large-scale distributed ML workloads~\cite{dean2012large,coates2013deep,chilimbi2014project,goyal2017accurate, pytorch, genet, hived, goyal, sun2019optimizing, mai, switch-ml, trio-ml}. Relevant to this paper, several papers aim to reduce communication overhead using smart  scheduling techniques; e.g., Gandiva~\cite{gandiva}, Themis~\cite{themis}, Pollux~\cite{pollux}, Tiresias~\cite{tiresias}, Shockwave~\cite{shockwave},  and Optimus~\cite{optimus}. These schedulers try to minimize network sharing by placing workers of the same job as close as possible to each other.
However, these approaches do not consider interleaving the communication patterns of different training jobs when placing them on servers. \name's contribution is complementary to these approaches by considering both the compute resources and the communication demands of different jobs during scheduling. Moreover, \name is designed as a pluggable module to augment these schedulers.

\begin{figure*}[t]
    \begin{minipage}{0.6\textwidth}
    \includegraphics[width=1.0\textwidth]{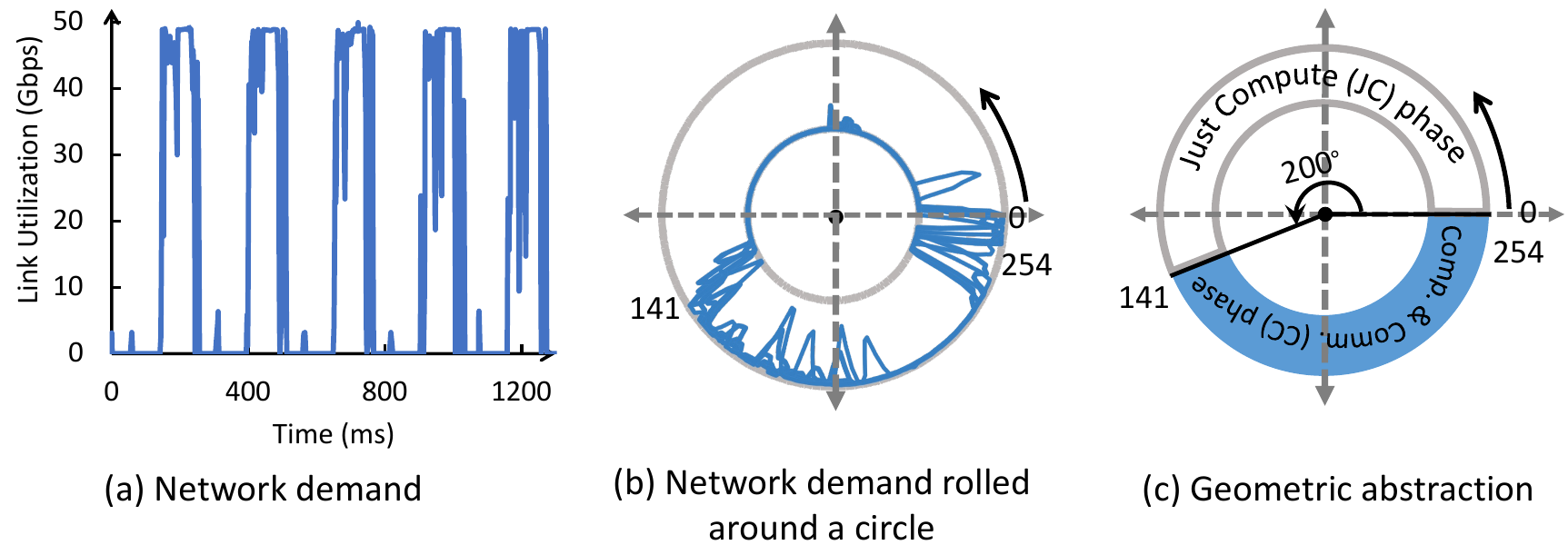}
    \caption{{\name's geometric abstraction}.}
    \label{fig:example_geometric}
    \end{minipage}
    \hspace{0.2cm}
    \begin{minipage}{0.35\textwidth}
    \includegraphics[width=1.0\textwidth]{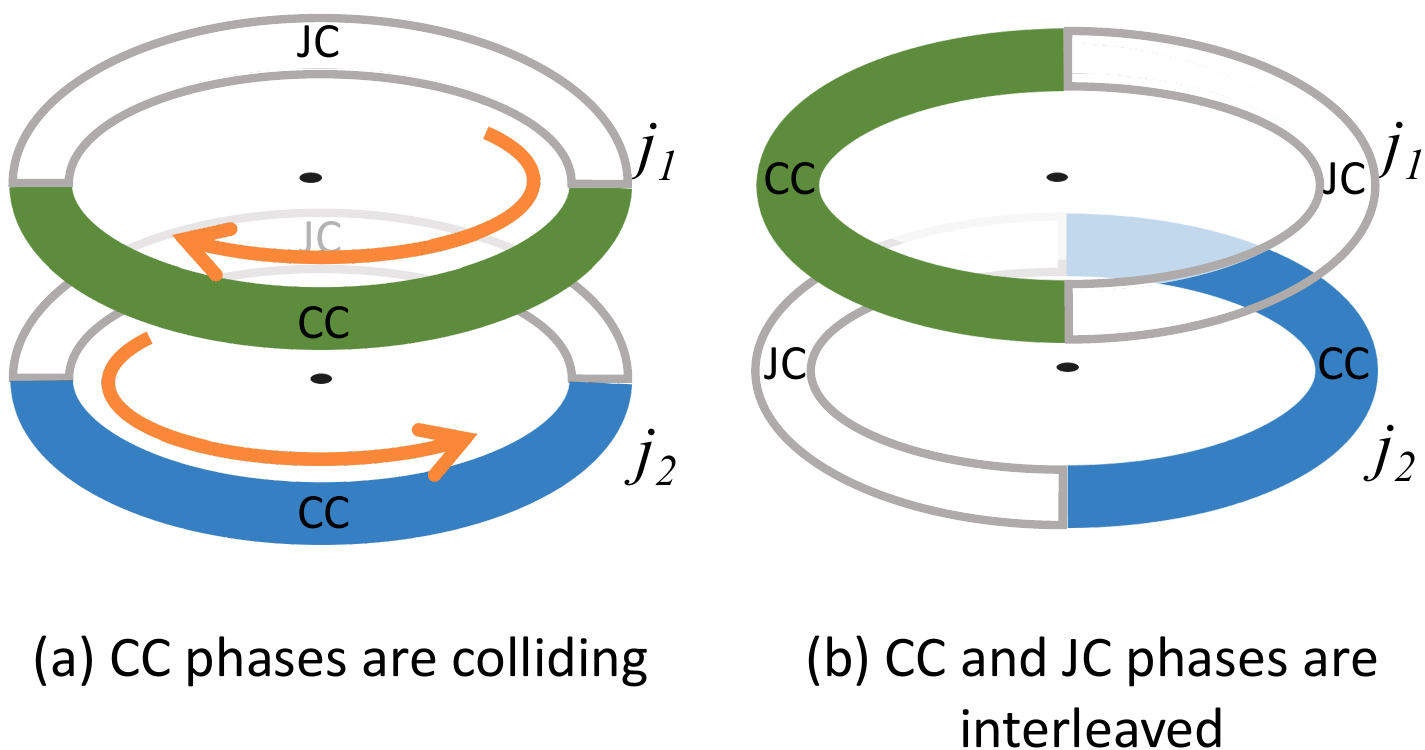}
    \caption{Rotating the circles enables interleaving the network demand of $j_1$ and $j_2$.}
    \label{fig:unfairness_to_geometric}
    \end{minipage}
\end{figure*}

\para{Multi-resource sharing.} Recently, Muri~\cite{muri} proposed a scheduling technique to interleave critical resources  (e.g., GPU, CPU, network, storage) of DNN training jobs. Muri packs jobs that are being executed \textit{on the same set of resources} into a group and interleaves their resource requirements using a Blossom-based scheduler. However, Muri's approach to resource interleaving only applies to jobs that share the same set of GPUs, CPUs, memory, and network resources.\footnote{Muri~\cite{muri} states this limitation: ``The algorithm avoids cross-group packing to minimize the packing overhead.''} Hence, Muri can interleave compute and communication phases of a set of jobs only if the jobs \textit{are sharing the same set of GPUs}. In contrast, \name is able to  interleave compute and communication phases of different jobs, irrespective of which GPUs they occupy. For instance, Muri's algorithm is not applicable to interleave the resources of $j_1$ and $j_2$ in Figure~\ref{fig:compatible_jobs}(a) because $j_1$ is distributed  between server$_1$ and server$_2$ while $j_2$ is distributed between server$_3$ and server$_4$; i.e., these two jobs do not belong to the same resource group in Muri's algorithm. Muri would have been able to interleave the resources if both $j_1$ and $j_2$ were distributed between all four servers. However, for many of the large models we use in our experiments, GPU-sharing is not possible because of the memory requirements of the model. Moreover, even with GPU sharing, in a large-scale cluster, cross-group network congestion is common. \name is able to interleave the \on and \off phases of different jobs, without requiring them to share the same set of resources. Similarly, Synergy~\cite{synergy_osdi22} proposed a multi-resource interleaving scheduling approach by inferring the sensitivity of DNN jobs to GPU, CPU, and memory resources using optimistic profiling. Synergy improves the overall cluster efficiency by performing resource-sensitive allocations instead of a GPU-proportional allocation. However, Synergy's approach does not consider the network bandwidth as a resource and is unable to interleave the communication phases with other resources. In contrast, \name's focus is on interleaving the network demand with GPU resources. \name is designed to augment both Muri and Synergy schedulers.

\para{Communication-aware scheduling.} A variety of approaches have been developed to accelerate communication among ML training workers of the same job to reduce network overhead~\cite{horovod, clopt, goyal2017accurate, baidu, blueconnect, wang2020blink, mudigere2021highperformance, sip-ml} and to enable more efficient pipelining strategies~\cite{gpipe, pipedream}. ByteScheduler~\cite{byteps_2} and Syndicate~\cite{syndicate} accelerate ML training by scheduling and optimizing the order of communication operations between different GPU servers used by a training job. ByteScheduler overlaps compute and communication operations \textit{within a training job}, while Syndicate provides a solution for planning and scheduling communication operations for large DNN training.
Similarly, TACCL~\cite{taccl}, BytePS~\cite{byteps_1}, and CLOPT~\cite{clopt} improve the communication collective algorithms of large DNN models. BytePS seeks to find a balance between the Parameter Server~\cite{osdi_parameter_server} and Ring-AllReduce algorithms for synchronizing the gradients. TACCL proposes a communication collective algorithm for training large models with data and model parallelism. CLOPT co-optimizes network topology and communication schedules for ML training. These approaches optimize communication \textit{within a training job}, however, they do not consider congestion and network sharing \textit{across training jobs}. In contrast, \name's approach is orthogonal to these techniques because \name focuses on sharing the network resources \textit{across different training jobs}.

\para{Difference with prior workshop paper.} A prior workshop paper~\cite{cassini_hotnets} introduced the idea of using a geometric abstraction to achieve job compatibility at a single-link level. We extend this workshop paper in a few important ways. First, the workshop paper considers compute/communication interleaving at a high level and does not provide a concrete scheduling technique to achieve it. Specifically, it relies on an unfair congestion control protocol to achieve interleaving but \name does not require any changes to or assumptions about the congestion control protocol. Second, the workshop paper ignores the impact of cluster-level interleaving. Third, the workshop paper only considers the data parallelism paradigm and its geometric abstraction does not generalize to model parallelism techniques. Finally, our optimization formulation, the \affinity graph abstraction, the design and implementation of the \name module, and our formal arguments around correctness (Theorem ~\ref{thm:uniqueness}), are all new contributions.

\section{Conclusion}
\name is a simple but effective approach that can integrate with existing cluster schedulers to allow them to accommodate multiple ML jobs' network needs. We introduce a novel metric, called compatibility score, to rank different GPU placements when jobs compete on network links. 
Our evaluations show that \name improves the average and tail completion time of jobs by up to 1.6$\times$ and 2.5$\times$, respectively. Moreover, we show that \name reduces the number of ECN marked packets by up to 33$\times$.

\label{bodypage}

\balance
\bibliographystyle{abbrv} 
\begin{small}
\bibliography{reference}
\end{small}

\appendix
\newpage

\section{Proof of Theorem~\ref{thm:uniqueness}}
\label{appendix:proof}

This section provides the proof of Theorem~\ref{thm:uniqueness} (Correctness and Uniqueness Guarantee) in Section~\ref{sec:design}. To prove uniqueness we need to show that Algorithm~\ref{alg:affinity} assigns a time shift value exactly once to every job $j\in J$ in the cluster with \affinity graph $G = (U,V,E)$. To prove the correctness, we need to show that:

\begin{align}
\forall l, \forall (j_{n}, j_{m}) \in  \{(j_{i}, j_{k}) & | (j_{i}, l) \in E \text{ and } (j_{k}, l) \in E\} : \nonumber \\
 (t_{j_{n}} - t_{j_{m}})\%p^l  =&  (t^{l}_{j_{n}} - t^{l}_{j_{m}})\%p^l \label{eq:correctness_contraint}
\end{align}
where $p^l$ is the perimeter of the geometric abstraction for link $l$. In other words, to guarantee correctness, we prove that for every pair of jobs sharing a link, the difference in the time-shift values assigned by the algorithm is equal to the relative time-shift given by \name's optimization formulation for that link.

We first use induction to prove that both the above statements are true for any connected and loop-free \affinity graph $G = (U_{1}, V_{1}, E_{1})$, and later we extend this to a general \affinity graph with many connected sub-graphs. \\

\textbf{Base case:} First, we show that both statements hold for a graph $G$ with only one link $l$. Algorithm~\ref{alg:affinity} first selects one of the jobs $j_{1}$ connected to the link $l$ and sets $t_{j_{1}} = 0$. Using its BFS traversal algorithm for all the other jobs $j_{i}$ connected to $l$, Algorithm~\ref{alg:affinity} sets the time shift as:
$$ t_{j_{i}} = - t^{l}_{j_{1}} + t^{l}_{j_{i}} $$
As the algorithm uses BFS and visits each job exactly once, the time-shift value is computed and assigned exactly once for each job. This ensures that for a given job, there is a unique time-shift value computed by the algorithm.

To show correctness, we need to prove equation~\ref{eq:correctness_contraint} for all job pairs connected to the link $l$. Say $j_{n}$ and $j_{m}$ are two jobs connected to the link $l$, then:

$$(t_{j_{m}} - t_{j_{n}})\%p^l = (- t^{l}_{j_{1}} + t^{l}_{j_{m}} - (- t^{l}_{j_{1}} + t^{l}_{j_{n}}))\%p^l $$
$$ = (t^{l}_{j_{m}} - t^{l}_{j_{n}})\%p^l $$
This shows that the time shift assignments are correct for the base case.

\textbf{Assumption Step:} Let us assume that the two statements hold for every connected and loop-free \affinity graph having $n$ links.

\textbf{Induction step:} We will use the above assumption to prove that the two statements hold for a connected and loop-free \affinity graph having $n+1$ links. Let $G=(U_{s}, V_{s}, E_{s})$ be the connected sub-graph with $n$ links, now we create an affinity graph with $n+1$ links, by adding a new link $l_{n}$ which is already connected to some set of jobs $J=\{j\}$. In order to get a connected and loop-free \affinity graph with $n+1$ links, $l_{n}$ has to be connected to exactly one job $j_{i}\in U_{s}$. It has to be exactly one because having an edge with more than one job from the sub-graph $G$ will form a loop, and not being connected with any of the jobs from the sub-graph $G$ will make the \affinity graph disconnected. Let $j_{i}$ be the job from subgraph $G$ that is connected to $l_{n}$. Since $j_{i}$ is the only path to reach the link $l_{n}$ and the jobs $J$ connected to the link, our algorithm~\ref{alg:affinity} will reach link $l_{n}$ through job $j_{i}$. Then from the algorithm, the time assignments for the jobs in $J$ are given by:

$$ \forall j_{k} \in J,  t_{j_{k}} = t_{j_{i}} - t^{l_{n}}_{j_{i}} + t^{l_{n}}_{j_{k}} $$

The uniqueness is guaranteed since BFS visits each job only once. From the assumption step the correctness constraints for all the links in the subgraph $G$ are assumed to be valid, we have to only prove equation~\ref{eq:correctness_contraint} for the jobs connected to $l_{n}$.

$$ \forall (j_{m}, j_{n}) \in J, (t_{j_{m}} - t_{j_{n}})\%p^l = (t_{j_{i}} - t^{l}_{j_{i}} + t^{l}_{j_{m}} - (t_{j_{i}} - t^{l}_{j_{i}} + t^{l}_{j_{n}}))\%p^l $$
$$ = (t^{l}_{j_{m}} - t^{l}_{j_{n}})\%p^l $$

This shows that both statements hold true for any \affinity graph with $n+1$ links. This concludes the induction proof, hence, Algorithm~\ref{alg:affinity} holds true for all connected and loop-free \affinity graphs.

Now we extend to an \affinity graph of a cluster having multiple connected sub-graphs. Since our algorithm solves each connected sub-graph one by one and assigns a single time-shift value for each job in the sub-graph, uniqueness is guaranteed. For correctness, since there is no edge connecting jobs and links from different disjoint sub-graphs there are no constraints across disjoint graphs that need to be checked for correctness. Hence this concludes the overall proof.

\begin{figure}[t]
    \centering
    \includegraphics[width=\columnwidth]{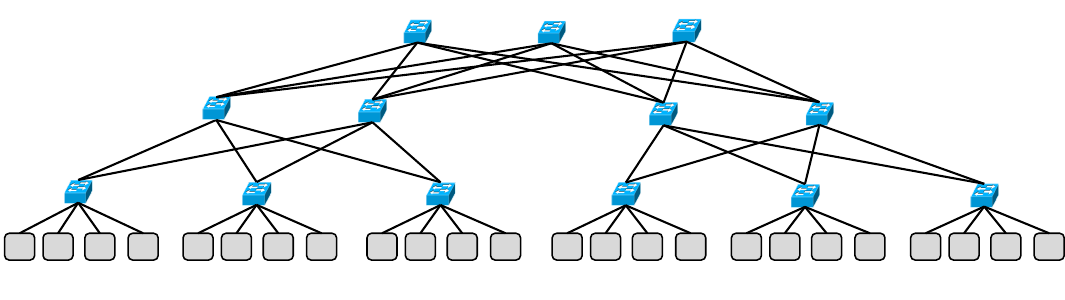}
    \caption{Logical topology of our testbed.}
    \label{fig:testbed_topology}
\end{figure}

\para{Example.} As an example, traversing the affinity graph in Figure~\ref{fig:affinity}(b) results in the following unique time-shifts for $j_1$, $j_2$, and $j_3$:

\vspace{-0.4cm}
\begin{align}
t_{j_1}=& 0~(\textit{reference point}) & \\ 
t_{j_2}=& (-t_{j_1}^{l_1} + t_{j_2}^{l_1}) & \mod{ &\hspace{1cm}\textit{iter\_time}_{j_2}} \\           & \textit{\affinity graph path:} & \textit{$j_1 \rightarrow$} &\textit{$l_1 \rightarrow j_2$} \nonumber\\
t_{j_3}=& (-t_{j_1}^{l_1} + t_{j_2}^{l_1} - t_{j_2}^{l_2} + t_{j_3}^{l_2}) & \mod{ & \hspace{1cm}\textit{iter\_time}_{j_3}}\\
         &\textit{\affinity graph path:} & \textit{$j_1 \rightarrow l_1$} & \textit{$\rightarrow j_2 \rightarrow l_2 \rightarrow j_3$}  \nonumber
\label{eq:case1}
\end{align}

For the correctness of the algorithm, the graph should be loop-free. In \name's design, we eliminate placement configurations that have loops. Themis allocates servers using an auction procedure, which involves multiple jobs in the cluster participating in the auction. This allows multiple possible placement configurations for the jobs participating in the auction. Hence it is easy to find many loop-free placement configurations among them. Similarly, Pollux reallocates resources periodically, involving multiple jobs and creating many possible placement configurations.

\section{Testbed Details}

\subsection{Topology}
\label{app:testbed_topology}

This section describes \name's testbed topology in more detail. Figure~\ref{fig:testbed_topology} illustrates the logical view of our testbed. Our testbed includes ASUS ESC4000A-E10 servers each with one A100 Nvidia GPU~\cite{a100} and one Mellanox ConnectX5 NIC. The NICs are connected to a Tofino switch. The Tofino switch emulates 13 logical switches and 48 bi-directional links for a 2:1 over-subscribed topology. 

\begin{table}
\scriptsize
\centering
\renewcommand{\arraystretch}{0.95}
\linespread{1.05}\selectfont\centering
\begin{tabular}{|p{1.9cm}|p{1.3cm}|p{1cm}|p{1.5cm}|p{0.9cm}|} 
\hline
DNN       &  Memory requirement (MB) & Batch size/GPU  & Parallelization strategy & Type \\\hline
VGG11~\cite{vgg11} &  507   &  512-1800  & Data Parallel & Vision  \\
VGG16~\cite{vgg16} &  528   &  512-1800  & Data Parallel & Vision  \\
VGG19~\cite{vgg19} &  549   &  512-1800  & Data Parallel & Vision  \\
WideResNet101~\cite{wideresnet} &   243  & 256-1200 & Data Parallel & Vision \\
ResNet50~\cite{resnet} &  98   &  256-1800  & Data Parallel & Vision  \\
BERT~\cite{bert}  &  450   &  8-32  & Data Parallel & Language \\
RoBERTa~\cite{roberta} &  800   &  8-32  & Data Parallel &Language \\
CamemBERT~\cite{camembert} &  266   &  8-32  & Data Parallel &  Language \\
XLM~\cite{xlm} &  1116   &  4-32  & Data Parallel & Language \\
GPT1~\cite{gpt_1} & 650 - 9000 &  32-80  & Model Parallel & Language \\
GPT2~\cite{gpt_2} & 1623- 27000 &  32-80  & Model Parallel & Language \\ 
GPT3~\cite{gpt_3} & 1952- 155000 &  16-48  & Model Parallel & Language \\ 
DLRM~\cite{dlrm}  &  890 - 1962  &  16-1024  & Model Parallel & Recomm. \\
\hline
\end{tabular}
\caption{DNN models used in our experiments.}
\label{tab:model_parameters}
\end{table}

\subsection{DNN Models}
\label{app:dnn_models_details}

As mentioned in Section~\ref{sec:eval_setup}, we run our experiment with 13 popular DNN models: VGG11~\cite{vgg11}, VGG16~\cite{vgg16}, VGG19~\cite{vgg19}, ResNet50~\cite{resnet}, WideResNet101~\cite{wideresnet}, BERT~\cite{bert}, RoBERTa~\cite{roberta}, XLM~\cite{xlm}, CamemBERT~\cite{camembert}, GPT1~\cite{gpt_1}, GPT2~\cite{gpt_2}, GPT3~\cite{gpt_3}, and DLRM~\cite{dlrm}. Table~\ref{tab:model_parameters} summarizes the parameters of each model and batch sizes. Note that the batch sizes are provided as a range because the number of workers and hyper-parameters change during scheduling epochs. In particular, in different experiments, we select the batch size according to the hyper-parameters used in prior work~\cite{measuring_effects_of_data_parallelism_2019, mlperf, switch-ml, themis, pollux}. The memory requirements of each model reflects the amount of memory each model occupies in the GPU memory. We adjust the model sizes for different models depending on the parallelization strategy.

In our experiments in Figure~\ref{fig:scatter_epoch}, we use different instances of the same model. We use suffixes on their names to distinguish between the instances, for example, GPT2-A and GPT2-B are two different training jobs, as shown in the legend of Figure~\ref{fig:scatter_epoch}. GPT2-A has a batch size of 24 with a model hidden size of 1536 (as defined by Deepspeed's codebase~\cite{deepspeed_gpt}), while GPT2-B has a batch size of 70 with a model hidden size of 1184. 

\section{Number of ECN Marked Packets}
\label{sec:ecn_marks}
\begin{figure}[t]
    \centering
    \includegraphics[width=\columnwidth]{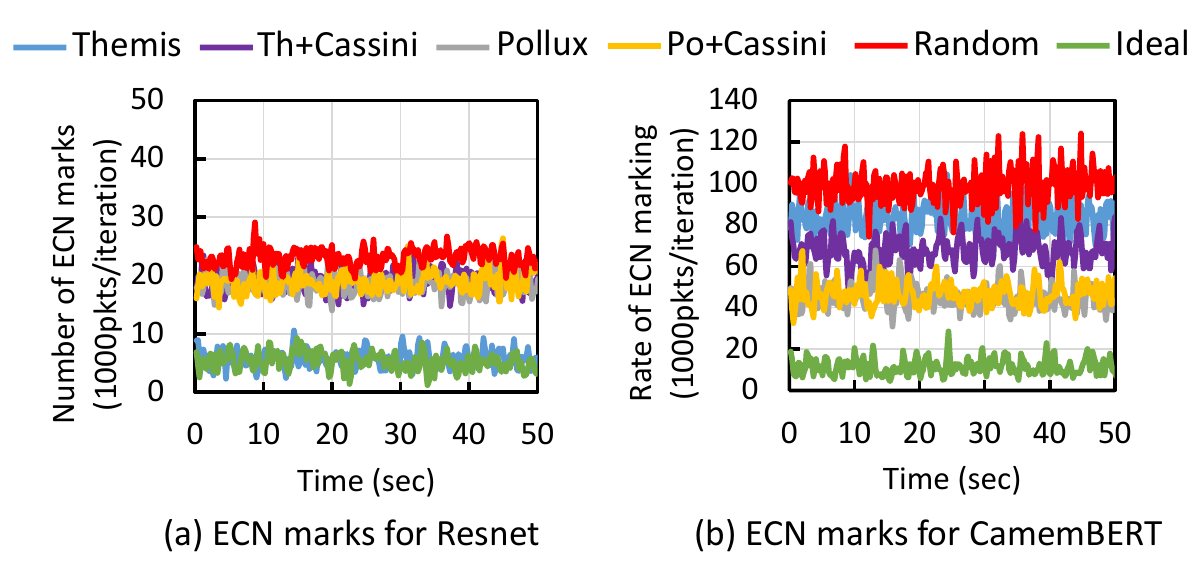}
    \caption{Number of ECN marked packets per iteration}
    \label{fig:ecn_marks_appendix}
\end{figure}

Figure~\ref{fig:ecn_marks_appendix} plots the number of ECN marked packets per iteration for the models ResNet and CamemBERT. These measurements are from the experiment of Section~\ref{sec:testbed_congestion_data_parallel}. The ResNet model has relatively lower ECN marks in general than other models because ResNet has a smaller model size and requires less network bandwidth for its AllReduce phase.

\end{document}